\documentclass{article}
\usepackage[T1]{fontenc}
\usepackage{amssymb, amsmath, amsthm, graphicx, subfigure}

\usepackage{float}
\usepackage[colorlinks=true, allcolors=blue]{hyperref}

\usepackage{xcolor}
\usepackage{braket}
\usepackage{esvect}
\usepackage{enumerate}
\usepackage{scalefnt}

\usepackage{pifont}% http://ctan.org/pkg/pifont
\newcommand{\cmark}{\ding{51}}%
\newcommand{\xmark}{\ding{55}}%
\newcommand{\omark}{\ding{109}}%

\newtheorem{theorem}{Theorem}[section]
\newtheorem{definition}[theorem]{Definition}
\newtheorem{proposition}[theorem]{Proposition}

\newtheorem{lemma}[theorem]{Lemma}

\newtheorem{remark}[theorem]{Remark}

\newtheorem{conjecture}{Conjecture}

\newcommand{\br}{\mathbf{r}}
\newcommand{\bv}{\mathbf{v}}

\newcommand{\RR}{\mathbb{R}}
\newcommand{\Exp}{\mathbb{E}}
\renewcommand{\Pr}{\mathbb{P}}

\newcommand{\ee}{{\rm e}}

\newcommand\smaller[2][0.85]{{\scalefont{#1}#2}}

\newcommand{\MAXCUT}{\mbox{\rm \smaller[0.76]{MAX CUT}}}
\newcommand{\MAXAND}{\mbox{\rm \smaller[0.76]{MAX 2-AND}}}
\newcommand{\MAXDICUT}{\mbox{\rm \smaller[0.76]{MAX DI-CUT}}}
\newcommand{\MAXSAT}[1]{\mbox{\rm \smaller[0.76]{MAX #1-SAT}}}
\newcommand{\MAXHORNSAT}[1]{\smaller[0.76]{MAX #1-HORN-SAT}}
\newcommand{\BMAXSAT}[1]{\smaller[0.76]{BALANCED MAX #1-SAT}}
\newcommand{\MAXSATg}{\smaller[0.76]{MAX SAT}}

\newcommand{\MAXCSP}[1]{\mbox{\rm \smaller[0.76]{MAX #1-CSP}}}
\newcommand{\MAXCSPF}[1]{\mbox{\rm \smaller[0.76]{MAX CSP(#1)}}}
\newcommand{\MAXLIN}[1]{\mbox{\rm \smaller[0.76]{MAX #1-LIN}}}
\newcommand{\MAXCSPP}{\mbox{\rm \smaller[0.76]{MAX CSP}}}
\newcommand{\VC}{\mbox{\rm \smaller[0.76]{VERTEX COVER}}}
\newcommand{\MAXkVC}{\mbox{\rm \smaller[0.76]{MAX $k$-VC}}}
\newcommand{\MINCUT}{\mbox{\rm \smaller[0.76]{MIN CUT}}}

\newcommand{\E}{\mathop{{}\mathbf{E}}}

\newcommand{\Val}{\mathrm{Val}}

\newcommand{\GE}{\;\ge\;}
\newcommand{\LE}{\;\le\;}
\newcommand{\EQ}{\;=\;}
\newcommand{\GT}{\;>\;}

\newcommand{\Value}{\mathsf{Value}}
\newcommand{\Prob}{\mathsf{Prob}}
\newcommand{\Ratio}{\mathsf{Ratio}}

\newcommand{\THRESH}{{\cal THRESH}}

\newcommand{\agw}{\alpha_{\text{GW}}}

\newcommand{\eps}{\varepsilon}

\begin{document}

\title{Tight approximability of MAX 2-SAT\\ and relatives, under UGC}
\date{}
\author{Joshua Brakensiek\thanks{Stanford University, supported in part by a Microsoft Research PhD Fellowship. Email: \texttt{jbrakens@cs.stanford.edu}}\and Neng Huang\thanks{University of Chicago, supported in part by NSF grant CCF:2008920. Email: \texttt{nenghuang@uchicago.edu}}\and Uri Zwick\thanks{Blavatnik School of Computer Science, Tel Aviv University, Israel. Email: \texttt{zwick@tau.ac.il}}}

\maketitle

\vspace*{-15pt}
\begin{abstract}
    Austrin showed that the approximation ratio $\beta\approx 0.94016567$ obtained by the MAX 2-SAT approximation algorithm of Lewin, Livnat and Zwick (LLZ) is optimal modulo the \emph{Unique Games Conjecture} (UGC) and modulo a \emph{Simplicity Conjecture} that states that the worst performance of the algorithm is obtained on so called \emph{simple} configurations. We prove Austrin's conjecture,  thereby showing the optimality of the LLZ approximation algorithm, relying only on the Unique Games Conjecture. Our proof uses a combination of analytic and computational tools. 
    
    We also present new approximation algorithms for two restrictions of the MAX 2-SAT problem. For MAX HORN-$\{1,2\}$-SAT, i.e., MAX CSP$(\{x\lor y,\bar{x}\lor y,x,\bar{x}\})$, in which clauses are not allowed to contain two negated literals, we obtain an approximation ratio of $0.94615981$. For MAX CSP$(\{x\lor y,x,\bar{x}\})$, i.e., when 2-clauses are not allowed to contain negated literals, we obtain an approximation ratio of $0.95397990$. By adapting Austrin's and our arguments for the MAX 2-SAT problem we show that these two approximation ratios are also tight, modulo only the UGC conjecture. This completes a full characterization of the approximability of the MAX 2-SAT problem and its restrictions.
\end{abstract}

\section{Introduction}

For over half a century \cite{J73}, computer scientists have been concerned with designing optimal approximation algorithms for a variety of optimization problems. One of the most popular classes of optimization problems to study is \emph{Boolean Constraint Satisfaction Problems}, where one is given a collection of Boolean \emph{variables} with \emph{constraints} (or predicates) on subsets of variables dictating what a valid assignment to the variables should satisfy. The optimization problem MAX CSP is concerned with finding a solution that maximizes the number of satisfied constraints. In general, this problem is NP-hard, but much research has gone into finding optimal approximation algorithms with corresponding hardness results.

The leading technique for approximating MAX CSPs is via semidefinite programming (SDP). This method was first used by Goemans and Williamson \cite{GW95} for a variety of MAX CSPs including \MAXCUT, \MAXSAT{2}, \MAXDICUT\, and \MAXSATg. In particular, the Goemans-Williamson algorithm for \MAXCUT\ achieves an approximation ratio of $\agw\approx 0.87856$, which is still the best known to date. The best evidence for the optimality of this algorithm is that, assuming Khot's Unique Games Conjecture (UGC)~\cite{khot02}, it is NP-hard to approximate \MAXCUT\ to a ratio of $\agw + \eps$ for any $\eps > 0$. This was proved by the combined efforts of Khot, Kindler, Mossel and O'Donnell \cite{KKMO07} who outlined the hardness reduction and Mossel, O'Donnell and Oleszkiewicz \cite{MOO10} (see also \cite{DeMoNe16}) who proved the \emph{Majority is Stablest} conjecture to complete the analysis.

However, the algorithms found by Goemans and Williamson for \MAXSAT{2}, \MAXDICUT\, and \MAXSATg\ have since been improved by Feige and Goemans \cite{FG95}, Matuura and Matsui \cite{MM03} and Lewin, Livnat and Zwick \cite{LLZ02} and others. See the survey of Makarychev and Makarychev \cite{MM17} and the recent progress by Brakensiek, Huang, Potechin and Zwick \cite{BHPZ22} for \MAXDICUT.

In this paper, we primarily focus on studying \MAXSAT{2}. Since the work of Lewin, Livnat and Zwick~\cite{LLZ02}, the best known approximation ratio for \MAXSAT{2} has been $\beta_{LLZ}\approx 0.940$. However, a proof of optimality has remained elusive. Khot, Kindler, Mossel and O'Donnell \cite{KKMO07} also showed that it is UGC-hard, i.e., NP-hard under the UGC conjecture, to approximate \BMAXSAT{2}, to within  $\beta_{BAL}+\eps$, for any $\eps>0$, where $\beta_{BAL}\approx 0.943943$ and \BMAXSAT{2} is the version of \MAXSAT{2} in which the total weight of clauses in which a variable appears positively is equal to the total weight of clauses in which it appears negatively. They also conjectured that balanced instances are the hardest for \MAXSAT{2}.

Austrin \cite{Austrin07}, refuting the conjecture of Khot, Kindler, Mossel and O'Donnell \cite{KKMO07} that balanced instances are hardest for \MAXSAT{2}, showed that it is UGC-hard to approximate \MAXSAT{2} to within $\beta^-_{LLZ}+\eps$, for any $\eps>0$, where $\beta^-_{LLZ}\approx 0.94016567$ is the ratio obtained by the LLZ algorithm with the optimal tuning of its parameters on \emph{simple}, but not necessarily balanced, configurations. Austrin \cite{Austrin07} conjectured that $\beta^-_{LLZ}=\beta_{LLZ}$, where $\beta_{LLZ}$ is the ratio obtained by the optimally tuned LLZ algorithm on \emph{all}, not necessarily simple, configurations, and hence on all instances of \MAXSAT{2}. We prove this conjecture, which we refer to as the \emph{Simplicity Conjecture}. (The conjecture is also implicit in \cite{LLZ02}.) This shows that it is UGC-hard to approximate \MAXSAT{2} to within $\beta_{LLZ}+\eps$, for any $\eps>0$, that the LLZ algorithm is essentially an optimal approximation algorithm for \MAXSAT{2}, and that $\beta_{LLZ}\approx 0.94016567$ is the exact approximability threshold of \MAXSAT{2}, under UGC. \footnote{Several previous papers, e.g., \cite{ChMaMa09,DeMo13,ODWu09,PoSc11}, state this result without mentioning that it relies on the now proven simplicity conjecture.}
\subsection{Our results}

Our main result is a proof of Austrin's simplicity conjecture. This establishes that the optimally tuned LLZ algorithm is indeed an optimal approximation algorithm for \MAXSAT{2}, relying only on UGC. More precisely, under UGC, for any $\eps>0$, it is NP-hard to approximate \MAXSAT{2} to within $\beta_{LLZ}+\eps$, where $\beta_{LLZ}\approx 0.94016567$ is the approximation ratio achieved by the optimally tuned LLZ algorithm. Interestingly, the only parameter used by the optimally tuned LLZ algorithm is actually~$\beta_{LLZ}$ itself. (See below.) There does not seem to be a closed form solution for the constant~$\beta_{LLZ}$ but it can be efficiently computed to any desired accuracy.

Although the simplicity conjecture is technical in nature, and by no means as profound as the Unique Games Conjecture, or the Majority is Stablest conjecture, turned theorem, we believe that proving it is important, since it implies that we do indeed know the tight approximability threshold of \MAXSAT{2}, modulo UGC. We note that the seemingly plausible and somewhat related conjecture of Khot, Kindler, Mossel and O'Donnell \cite{KKMO07} that balanced instances of \MAXSAT{2} are the hardest turned out to be false. We also note that the family of rounding functions believed to be sufficient for obtaining an optimal approximation algorithms for \MAXAND\ and \MAXDICUT\ turned out not to be sufficient, as shown in \cite{BHPZ22}.

We also provide a full classification of all possible restrictions of the \MAXSAT{$\{1,2\}$}, i.e., the \MAXSATg\ problem with clauses of sizes~$1$ and~$2$. It turns out that up to symmetries there are four non-equivalent restrictions  two of which are nontrivial. (See Table~\ref{T-class}.)

Austrin's hardness proof for \MAXSAT{2} only uses clauses of the form $x \vee y$ and $\bar{x} \vee \bar{y}$, so any restriction that contains these two types is automatically as hard as \MAXSAT{2} itself.

The first non-trivial restriction of \MAXSAT{$\{1,2\}$} is \MAXHORNSAT{$\{1,2\}$} in which clauses of the form $\bar{x}\lor \bar{y}$ are not allowed. We show that a noticeably better approximation ratio of about $0.94615981$ can be obtained for this problem and that this is tight. A further noticeable improvement, to about $0.95397990$, is obtained if clauses of the forms $\bar{x}\lor y$ and $\bar{x}\lor \bar{y}$ are not allowed. This is again tight.

The problem \MAXCSPF{$\{x\lor y,\bar{x}\}$} may be seen as a variant of the \VC\ problem. Given an undirected graph $G=(V,E)$, choose a subset $A\subseteq V$ so as to maximize the number of edges covered by~$A$, plus the number of vertices not in~$A$. It is possible to add nonnegative weights to the vertices and edges. (A related problem, \MAXkVC, is mentioned in the concluding remarks.)

The remaining two restrictions of \MAXSAT{$\{1,2\}$} can be solved exactly in polynomial time. The problem \MAXCSPF{$\{\bar{x}\lor y,x,\bar{x}\}$} can be solved exactly via a reduction to the $s$-$t$ \MINCUT\ problem in directed graphs. Instances of \MAXCSPF{$\{x\lor y, \bar{x}\lor y,x\}$} are always satisfied by the all-$1$ assignment.

\begin{table}[t]
\begin{center}
\begin{tabular}{ccccccc}
Name & \makebox[1cm][c]{$x\lor y$} & \makebox[1cm][c]{$\bar{x}\lor y$} & \makebox[1cm][c]{$\bar{x}\lor \bar{y}$} & \makebox[1cm][c]{$x$} & \makebox[1cm][c]{$\bar{x}$} & Approximation Ratio \\
\hline
\vphantom{$2^{2^{2^2}}$}
\MAXSAT{$\{1,2\}$} & \cmark & \omark & \cmark & \omark & \omark & $\approx0.94016567$ \\ 
\MAXHORNSAT{$\{1,2\}$} & \cmark & \cmark & \xmark & \omark & \cmark & $\approx0.94615981$ \\
\MAXCSPF{$\{x\lor y,x,\bar{x}\}$}& \cmark & \xmark & \xmark & \omark & \cmark & $\approx 0.95397990$ \\[3pt]
\hline
\vphantom{$2^{2^{2^2}}$}
\MAXCSPF{$\{\bar{x}\lor y,x,\bar{x}\}$}& \xmark & \omark & \xmark & \omark & \omark & $1$ \\
\MAXCSPF{$\{x\lor y, \bar{x}\lor y,x\}$}& \omark & \omark & \xmark & \omark & \xmark & $1$ \\[3pt]
\hline
\end{tabular}
\end{center}
\caption{\MAXSAT{$\{1,2\}$} and its non-equivalent subproblems. \cmark\ indicates that the use of clauses of the corresponding type is allowed. \xmark\ indicates that the use of such clauses is not allowed. \omark\ indicates that using or not using such clauses does not change the approximation ratio of the problem.}\label{T-class}
\end{table}

\subsection{Significance of the results}

As mentioned, we believe that proving Austrin's simplicity conjecture is important as it gives a tight approximability result for \MAXSAT{2}, relying only on UGC. It also enhances our understanding of the LLZ algorithm. An interesting consequence is that the rounding procedure needed to obtain an optimal approximation algorithm for \MAXSAT{2} is  much simpler than the ``universal'' rounding procedures used by Raghavendra's \cite{R08,R09} result. (See Section~\ref{sub-related}.) In particular, only one Gaussian random variable is needed. 

Studying restrictions of the \MAXSAT{2} problem is interesting as it provides more natural constraint satisfaction problems for which tight approximation results are now known. What is also striking is that the rounding functions used for \MAXSAT{2} and its subproblems are especially simple and clean. This stands in sharp contrast to the very complicated rounding functions needed to obtain close to optimal approximation algorithms for \MAXAND\ and \MAXDICUT.

\subsection{Techniques}

The LLZ algorithm itself is fairly simple and natural. It is parameterized by a \emph{threshold} function $f:[-1,1]\to[-1,1]$, or by a distribution $\cal F$ of threshold functions. The threshold function used to obtain an optimal approximation algorithm for \MAXSAT{2} is very simple: $f(x)=\beta x$, where $\beta=\beta_{LLZ}$ is the optimal approximation ratio of \MAXSAT{2}.

The analysis of the LLZ algorithm, however, requires minimizing or maximizing functions of several real variables that involve the cumulative probability function of 2-dimensional Gaussian variables for which no closed form exists. This greatly complicates the analysis. All previous analyses of the LLZ algorithm had to resort to numerical techniques. Lewin, Livnat and Zwick \cite{LLZ02} used non-rigorous numerical techniques to obtain the estimate $\beta_{LLZ}\approx 0.940$. Austrin \cite{Austrin07} used non-rigorous numerical techniques to corroborate his simplicity conjecture. Sj{\"o}gren \cite{Sjogren09} used rigorous numerical techniques to show that $\beta_{LLZ}\ge 0.94016$, but did not attempt to prove the simplicity conjecture.

To prove the simplicity conjecture we need to obtain a much finer analysis of the LLZ algorithm. We are not interested in just lower bounding the approximation ratio obtained by the algorithm. We need to show that the \emph{critical configurations}, i.e., the configurations on which the worst performance of the algorithm is obtained have a particularly simple form. We do that using a combination of analytic techniques and rigorous numerical techniques, i.e., \emph{interval arithmetic}.

Interval arithmetic was previously used by Zwick \cite{zwick02}, Sj{\"o}gren \cite{Sjogren09}, Austrin, Benabbas and Georgiou~\cite{AuBeGe16}, Bhangale et al.~\cite{BhangaleKKST18}, Bhangale and Khot \cite{BhKh20}, and Brakensiek et al.~\cite{BHPZ22} to analyze SDP-based approximation algorithms. Our use of interval arithmetic is more involved since we are not just trying to lower bound the approximation ratio achieved by an algorithm, but rather certify that the worst behavior is obtained on configurations of a certain form, the so called simple configurations. This requires a much more careful analysis.

\subsection{Comparison to the works of Raghavendra and Austrin}

The seminal work of Raghavendra~\cite{R08} showed, assuming the Unique Games Conjecture, that every \MAXCSPP, including the ones studied in this paper, have a sharp approximation threshold $\alpha \in [-1, 1]$ such that for every $\eps > 0$, there exists an efficient SDP rounding algorithm achieving an $\alpha - \eps$ approximation and that no $\alpha + \eps$ approximation algorithm exists assuming the unique games conjecture. Further, the follow-up work of Raghavendra and Steurer~\cite{RS09} showed that an explicit SDP integrality gap for $\alpha+\eps$ can be computed in $O(\exp(\exp(1/\eps)))$ time. 

We observe that Raghavendra's theorem does not shed much insight on Austrin's simplicity conjecture. The reason is that the rounding functions considered by Raghavendra (as well as in his revised proof with Brown-Cohen~\cite{BCR15}) are very complex. In particular, after SDP vectors $\bv_i$ are computed for every variable in the CSP (and a reference vector $\bv_0$), one samples $N=N(\eps)$ $n$-dimensional Gaussian random variables $\br_1, \hdots, \br_N$, and rounds variable $x_i$ based on the values $\bv_i \cdot \bv_0, \bv_i \cdot \br_1, \hdots, \bv_i \cdot \br_N$. This is in contrast to the rounding algorithms considered by LLZ and Austrin which only require one randomly sampled $n$-dimensional Gaussian random variable (this is known as $\THRESH^-$ or $\THRESH$, see Section~\ref{sec:thresh}). Thus, for the variants of \MAXSAT{2} problem, we go beyond Raghavendra's theorem by showing that simple families of rounding functions are exactly optimal (assuming the Unique Games Conjecture). We hope the approximation community takes interest in improving beyond Raghavendra's theorem not just for 2-CSPs but for \MAXCSPP s in general.

In addition, the work of Austrin~\cite{Austrin10} on \MAXAND\ is closely related to our main results. In particular, he shows that for any \MAXCSP{2}, a probability distribution of SDP vectors which is difficult to round with any $\THRESH^-$ rounding scheme can be converted into an UGC-hardness proof, see Section~\ref{sec:PCP} for more details. However, Austrin imposes a \emph{positivity} conditions on these vectors for the hardness proof to go through. As such, it is currently an open question if every \MAXCSP{2} can be tightly captured by Austrin's hardness framework. We show that for all the problems which we study in this paper--\MAXSAT{2}, \MAXCSPF{$(\{x\lor y,x,\bar{x}\})$}, and \MAXHORNSAT{$\{1,2\}$}--their hardness analyses do follow from Austrin's hardness framework (or more precisely a mild generalization of the result from \cite{BHPZ22} which allows for \MAXCSPP s with non-negated literals). However in each case, we require considerable effort to tightly analyze a matching algorithm,\footnote{The hard distribution and matching algorithm for \MAXCSPF{$(\{x\lor y,x,\bar{x}\})$}, and \MAXHORNSAT{$\{1,2\}$} were found by adapting the code/methods of \cite{BHPZ22}.} which is the primary novelty of our paper.

\subsection{Other Related Work}\label{sub-related}

The literature for the MAX CSP problem and its approximation algorithms is extremely broad. See Makarychev and Makarychev \cite{MM17} for a contemporary survey. We highlight a few other works closely related to our investigation.

Discounting the unique games conjecture, all the best known NP-hardness results for approximating MAX CSPs are based on the PCP theorem~\cite{ALMSS98}. For \MAXSAT{2} the best known hardness ratio is $\frac{21}{22}+\eps$ for all $\eps > 0$ due to H{\aa}stad \cite{H01}. This paper also proved tight hardness results for the \MAXSAT{3} ($7/8+\eps$) and \MAXLIN{3} ($1/2+\eps$) problems.

For Boolean \MAXCSPP s, Creignou~\cite{C95} shows that every such problem can be solved exactly or there exists a constant $\eps > 0$ such that obtaining an approximation ratio of $1-\eps$ is NP-hard. In particular, the problems \MAXCSPF{$\{\bar{x}\lor y,x,\bar{x}\}$} and \MAXCSPF{$\{x\lor y, \bar{x}\lor y,x\}$} were already known to be exactly solvable. See also the generalization due to Khanna, Sudan, Trevisan and Williamson~\cite{KSTW01}.

The work of Thapper and {\v{Z}}ivn{\`y}~\cite{TZ16} gives a classification of which \MAXCSPP s (or, more precisely, a more general family known as Valued CSPs) can be solved exactly on any domain.

\subsection*{Organization}

In Section~\ref{sec:prelim}, we formally describe $\MAXCSPP$, the SDP relaxation, rounding algorithms, and techniques for showing UGC-hardness. In Section~\ref{section:2SAT}, we prove the Simplicity Conjecture. In Section~\ref{sec:ORNOT}, we give a tight analysis of \MAXCSPF{$(\{x\lor y,x,\bar{x}\})$}. In Section~\ref{sec:horn}, we tightly analyze \MAXHORNSAT{$\{1,2\}$}. In Section~\ref{S-concl}, we leave concluding remarks and open problems. In Appendix~\ref{app:IA}, we describe the implementation details of the interval arithmetic verification. In Appendix~\ref{app:beta}, we describe an explicit formula for $\beta_{LLZ}$, the optimal \MAXSAT{2} approximation ratio.

\section{Preliminaries}\label{sec:prelim}

\subsection{The MAX 2-SAT problem and its relatives}

A Boolean predicate of arity $k$ is a function $P: \{-1, 1\}^k \to \{0, 1\}$, where in the domain we associate~$-1$ with true and~$1$ with false. We say that $P$ is satisfied by $x \in \{-1, 1\}^k$ if $P(x) = 1$.

\begin{definition}[$\MAXCSPP(\mathcal{P})$]\label{def:csp+}
    Let $\mathcal{P}$ be a set of Boolean predicates. An instance of $\MAXCSPP(\mathcal{P})$ is defined by the following:
    \begin{itemize}
        \item A set of Boolean variables $\mathcal{V} = \{x_1, \ldots, x_n\}$.
        \item A set of constraints $\mathcal{C} = \{C_1, \ldots, C_m\}$, where each constraint $C_i = P_i(x_{j_{i, 1}}, x_{j_{i, 2}}, \ldots, x_{j_{i, k}})$ for some $P_i \in \mathcal{P}$ with arity $k$ and $j_{i, 1}, \ldots, j_{i, k} \in [n]$.
        \item A weight function $w: \mathcal{C} \to [0, 1]$ satisfying $\sum_{i = 1}^m w(C_i) = 1$.
    \end{itemize}
    The goal is to find an assignment to the variables that maximizes $\sum_{i = 1}^m w(C_i) \cdot P_i(x_{j_{i, 1}}, x_{j_{i, 2}}, \ldots, x_{j_{i, k}})$, the sum of the weights of satisfied constraints.
\end{definition}

\begin{definition}
    The $\MAXSAT{2}$ problem is the problem $\MAXCSPP(\mathcal{P}_{2SAT})$, where $\mathcal{P}_{2SAT} = \{x \vee y, x \vee \bar{y}, \bar{x} \vee y, \bar{x} \vee \bar{y}, x, \bar{x}\}$.
\end{definition}

We remark that Definition~\ref{def:csp+} does not allow negated variables. This paper will focus on the case where $\mathcal{P}$ is a subset of $\mathcal{P}_{2SAT}$, in other words, $\MAXCSPP(\mathcal{P})$ is  \MAXSAT{2} or one of its subproblems.

For the special case where every predicate in $\mathcal{P}$ is of arity at most 2, we can write down the \emph{canonical} SDP relaxation as follows (c.f., \cite{BHPZ22}).

\begin{alignat*}{3}
&\text{maximize} &\qquad& \sum_{C=P(x_{i}, x_{j}) \in \mathcal C}  w(C) \cdot \left( \hat{P}_\emptyset + \hat{P}_i\bv_0\cdot\bv_{i} +\hat{P}_j \bv_0\cdot\bv_{j} + \hat{P}_{i,j} \bv_{i}\cdot\bv_{j}\right)\\
&\text{subject to} &\qquad& \forall i \in \{0, 1, 2, \ldots, n\}\;,\ \ \ \qquad\,\, \bv_i \cdot \bv_i = 1\;,\\
& &\qquad& \forall C = P(x_{i}, x_{j}) \in \mathcal{C}\;, \quad
\begin{array}{c}
(\bv_0 - \bv_i) \cdot (\bv_0 - \bv_j) \GE 0\;,\\
(\bv_0 + \bv_i) \cdot (\bv_0 - \bv_j) \GE 0\;,\\
(\bv_0 - \bv_i) \cdot (\bv_0 + \bv_j) \GE 0\;,\\
(\bv_0 + \bv_i) \cdot (\bv_0 + \bv_j) \GE 0\;.\\
\end{array}
\end{alignat*}

Here, for every variable $x_i$ in the CSP instance, we have a vector-valued variable $\bv_i$, plus a special vector $\bv_0$ representing the truth value false. These vectors are all constrained to be unit vectors. The last four constraints are called \emph{triangle inequalities}, which clearly hold for integer solutions. In the objective value we take the \emph{Fourier expansion} $P(x_i, x_j) = \hat{P}_\emptyset + \hat{P}_ix_i +\hat{P}_j x_j + \hat{P}_{i,j} x_ix_j$ and replace the Boolean monomials with vector inner products. This relaxation is called canonical since any integrality gap instance of this relaxation can be turned into an NP-hardness result, assuming the famous Unique Games Conjecture.

\begin{definition}[Unique Games, as stated in \cite{BHPZ22}]
    An instance $I = (G, L, \Pi)$ of unique games is specified by a weighted graph $G = (V(G), E(G), w)$, a finite set of labels $[L] = \{1, \ldots, L\}$ and a set of permutations $\Pi = \{\pi_e^u : [L] \to [L] \mid e = \{u, v\}\in E(G)\}$ such that for every edge $e = \{u, v\}\in E(G)$, the two permutations $\pi_e^u$ and $\pi_e^v$ are inverses of each other, i.e., $\pi_e^u = (\pi_e^v)^{-1}$. For any assignment $A: V(G) \to [L]$, we say that $A$ satisfies an edge $e = \{u, v\}$ if $\pi_e^u(A(u)) = A(v)$, and we define $\Val(I, A) = \sum_{e \in E(G): A \text{ satisfies } e} w(e)$ to be the total weights of edges satisfied by $A$. We define the value of the instance $\Val(I)$ to be the value of the best assignment, i.e, $\Val(I) = \max_A \Val(I, A)$.
\end{definition}

\begin{conjecture}[Unique Games Conjecture, as stated in \cite{BHPZ22}]
    For every $\eta, \gamma > 0$, there exists a sufficiently large $L$ such that, given a unique games instance $I$ with $L$ labels, it is NP-hard to determine if $\Val(I) \geq 1 - \eta$ or $\Val(I) \leq \gamma$. 
\end{conjecture}

We say that a problem is UG-hard, if it is NP-hard assuming the Unique Games Conjecture. Raghavendra~\cite{R08} showed a generic reduction converting any integrality gap instance of the canonical SDP relaxation into a UG-hardness result.

\subsection{Configurations of biases and pairwise biases}

To analyze SDP-based approximation algorithms for CSPs, the common approach is to show that for every constraint, the probability that the algorithm satisfies the constraint is at least some ratio times the relaxed value of the constraint. For such analysis, it is convenient to consider \emph{configurations} of biases and pairwise biases, which are tuples of inner products between SDP vectors that appear in the same constraint. For constraints involving two variables, such configurations are triplets of the form $(b_i, b_j, b_{ij})$, where $b_i = \bv_0 \cdot \bv_i$ and $b_j = \bv_0 \cdot \bv_j$ are the \emph{biases} and $b_{ij} = \bv_i \cdot \bv_j$ is the \emph{pairwise bias}. For a single-variable constraint, the configuration is of the form $(b_i)$, a single bias. Each configuration also has a \emph{predicate type}, which is the predicate used to define the constraint. We will often use $\theta$ to denote a configuration and use the term $k$-configuration to refer to configurations with a predicate type of arity $k$. In this work, we will often implicitly assume that a $2$-configuration is of predicate type $x \vee y$, unless otherwise specified. We say that a configuration is \emph{feasible} if it appears in a feasible SDP solution.

\begin{definition}[Relative pairwise bias]
Given a 2-configuration $\theta = (b_i, b_j, b_{ij})$, the relative pairwise bias is defined as $\rho(\theta)=\frac{b_{ij} - b_ib_j}{\sqrt{(1 - b_i^2)(1 - b_j^2)}}$, if $(1 - b_i^2)(1 - b_j^2)\neq 0$, and $0$ otherwise.
\end{definition}
Geometrically, $\rho(\theta)$ is the renormalized inner product between $\bv_i$ and $\bv_j$ after removing their components parallel to $\bv_0$. 

\begin{definition}[Positive configurations~\cite{Austrin10}]
A 2-configuration $\theta = (b_i, b_j, b_{ij})$ with predicate type~$P$ is called \emph{positive} if $\hat{P}_{i,j}\cdot \rho(\theta) \geq 0$.
\end{definition}

Austrin conjectured that positive configurations are the hardest configurations for 2-CSPs, and showed how to obtain UG-hardness results from hard distributions of positive configurations~\cite{Austrin10}. Specifically for \MAXSAT{2}, he also conjectured that the hardest configurations are of the following even simpler form.

\begin{definition}[Simple configurations~\cite{Austrin07}]
    A 2-configuration $\theta = (b_i, b_j, b_{ij})$ is called simple if we have $b_i = b_j = b$ and $b_{ij} = -1 + 2|b|$ for some $b \in [-1, 1]$.    
\end{definition}

For the predicate type $P = x_i \vee x_j$, any simple configuration $\theta$ is also positive, since we have $\rho(\theta) = \frac{-1 + 2|b| - b^2}{1 - b^2} = - \frac{(1 - |b|)^2}{1 - b^2} \leq 0$ and $\hat{P}_{ij} = -1/4 < 0$ as well.

\begin{definition}
    Given a configuration $\theta$, we define $\Value(\theta)$ to be its relaxed value in the SDP. Given a distribution of configurations $\Theta$, we define its SDP value to be $\Value(\Theta) = \E_{\theta \sim \Theta} \Value(\theta)$.
\end{definition} 
We remark that for a $2$-configuration $\theta = (b_i, b_j, b_{ij})$ with predicate type $x_i \vee x_j$, we have
\[
\Value(\theta) \EQ \Value(b_i, b_j, b_{ij}) \EQ \frac{3 - b_i - b_j - b_{ij}}{4}\;.
\]

\subsection{1-dimensional and 2-dimensional normal random variables}

We collect a few facts about normal random variables and basic calculus, which will be important in our analysis of the algorithms. Let $\varphi(x)=\frac{1}{\sqrt{2\pi}}\ee^{-x^2/2}$ and $\Phi(x)=\int_{-\infty}^x \varphi(t)dt$ be the probability density function and cumulative probability function of the standard normal random variable $X\sim N(0,1)$, i.e., $\Phi(x)=\Pr[X\le x]$. Note that $\Phi'(x)=\varphi(x)$ and $\varphi'(x)=-x\varphi(x)$. Let 
\[\varphi_\rho(x,y) \EQ \frac{1}{2\pi\sqrt{1-\rho^2}}\ee^{-\frac{x^2-2\rho x y+y^2}{2(1-\rho^2)}}\] 
be the probability density function of a pair $(X,Y)$ for standard normal variables with correlation $\Exp[XY]=\rho$, where $-1< \rho < 1$. (Note that $\varphi_0(x,y)=\varphi(x)\varphi(y)$.) The cumulative distribution function of $(X,Y)$ is then:
\[\Phi_\rho(x,y) \EQ \Phi(x,y,\rho) \EQ \Pr[X\le x \land Y\le y] \EQ \int_{-\infty}^x\int_{-\infty}^{y}\varphi_\rho(t_1,t_2)dt_1 dt_2\;.\]

\begin{lemma}\label{L:phi_rho}
The partial derivatives of $\Phi(x,y,\rho)=\Phi_\rho(x, y)$ are:
\begin{align*}
\frac{\partial \Phi(x, y,\rho)}{\partial x} & \EQ \varphi(x)\Phi\left(\frac{y - \rho x}{\sqrt{1 - \rho^2}}\right)\;, \\
\frac{\partial \Phi(x, y,\rho)}{\partial y} & \EQ \varphi(y)\Phi\left(\frac{x - \rho y}{\sqrt{1 - \rho^2}}\right)\;, \\
\frac{\partial \Phi(x, y,\rho)}{\partial \rho} & \EQ \frac{1}{2\pi \sqrt{1 - \rho^2}} \exp\left(-\frac{x^2 - 2\rho xy + y^2}{2(1 - \rho^2)}\right)\;.
\end{align*}
\end{lemma}
The last equation above is Equation 4 from~\cite{DW90}.

We recall the following trivial facts.

\begin{proposition}\label{prop:Phi_diff}
    For any $a < b$, we have
    \[
    \Phi(b) - \Phi(a) \GE (b - a) \cdot \min(\varphi(a), \varphi(b))\;.
    \]
\end{proposition}
\begin{proof}
    We have
    \[
    \Phi(b) - \Phi(a) \EQ \int_a^b \varphi(t) dt \GE (b - a) \cdot \min_{t \in [a, b]}\varphi(t)\;.
    \]
    The proposition then follows from the fact that $\varphi(t) = \frac{1}{\sqrt{2\pi}}\exp(-t^2/2)$ achieves its minimum on any interval $[a, b]$ on the endpoints.
\end{proof}

The following proposition is taken from Austrin \cite{Austrin07} where a simple proof can also be found.
\begin{proposition}\label{prop:bivariate_symmetry}
    For all $x, y \in \mathbb{R}$, $\rho \in [-1, 1]$, we have
    \[
        \Phi_{\rho}(x, y) - \Phi_{\rho}(-x, -y) \EQ \Phi(x) + \Phi(y) - 1\;.
    \]
\end{proposition}

\subsection{The $\THRESH^-$  rounding family and the LLZ algorithm for \MAXSAT{2}}\label{sec:thresh}

The $\THRESH^-$ rounding family~\cite{LLZ02} is the following rounding procedure for the canonical SDP relaxation. Let $\bv_0,\bv_1,\ldots,\bv_n$ be a vector solution obtained by solving the canonical SDP, and let $b_i = \bv_0 \cdot \bv_i$ and $b_{ij} = \bv_i \cdot \bv_j$ be the biases and pairwise biases. We define 
\[
\bv_i^\perp \EQ \frac{\bv_i - b_i \cdot \bv_0}{\sqrt{1 - b_i^2}}
\]
to be the component of $\bv_i$ that's orthogonal to $\bv_0$, renormalized to a unit vector (if $|b_i| = 1$, we can take $\bv_i^\perp$ to be a unit vector that's orthogonal to every other SDP vector). A rounding algorithm from the $\THRESH^-$ family is specified by a \emph{threshold} function $f: [-1, 1] \to [-1,1]$. The algorithm chooses a standard normal random vector $\br$ and for every $i$, sets the variable $x_i$ to true if and only if $\bv_i^\perp \cdot \br \geq \Phi^{-1}(\frac{1+f(b_i)}{2})$. Since $\br$ is a standard normal random vector, we have $\bv_i^\perp \cdot \br \sim N(0, 1)$ for every variable $x_i$, and furthermore the correlation between $\bv_i^\perp \cdot \br$ and $\bv_j^\perp \cdot \br$ is equal to $\bv_i^\perp \cdot \bv_j^\perp$. Note that for a 2-configuration $\theta$ given by the SDP vectors $\bv_i, \bv_j$, we have exactly $\rho(\theta) = \bv_i^\perp \cdot \bv_j^\perp$.

The LLZ algorithm for \MAXSAT{2} is an algorithm from the $\THRESH^-$ family. In particular, it chooses $f: b \mapsto \beta b$ for some parameter $\beta \leq 1$. 

\begin{definition}
    Given a configuration $\theta$, we define $\Prob_\beta(\theta)$ to be the probability that it's satisfied by the LLZ algorithm with parameter $\beta$. Given a distribution of configurations $\Theta$, we define $\Prob_\beta(\Theta)$ to be $\E_{\theta \sim \Theta}\Prob_\beta(\theta)$.
\end{definition}

For a 2-SAT configuration $\theta = (b_i, b_j, b_{ij})$, the probability that it's satisfied by the LLZ algorithm is equal to

\[ \Prob_\beta(\theta) \EQ \Prob_\beta (b_i, b_j, b_{ij}) \EQ 1 - \Phi_{\rho(\theta)}\left(\Phi^{-1}\left(\frac{1+\beta b_i}{2}\right),\Phi^{-1}\left(\frac{1+\beta b_j}{2}\right)\right), \]
since the probability that both variables are set to false is exactly $\Phi_{\rho(\theta)}\left(\Phi^{-1}\left(\frac{1+\beta b_i}{2}\right),\Phi^{-1}\left(\frac{1+\beta b_j}{2}\right)\right)$. The ratio achieved by the algorithm on this configuration is then equal to
\[ 
\Ratio_\beta (b_i, b_j, b_{ij}) \EQ \frac{\Prob_\beta (b_i, b_j, b_{ij})}{\Value(b_i, b_j, b_{ij})}\;,
\]
and the approximation ratio achieved by the algorithm overall can be obtained by taking the minimum over all feasible configurations $(b_i, b_j, b_{ij})$. We then choose $\beta$ to maximize this ratio.

Lewin, Livnat and Zwick \cite{LLZ02} and Austrin \cite{Austrin07} made the following conjecture:

\begin{conjecture}[Simplicity Conjecture]\label{conj:Austrin}
    The worst performance ratio of the optimized LLZ algorithm is obtained on simple configurations, i,e., configurations of the form $(b, b, -1 + 2|b|)$.
\end{conjecture}

Modulo this conjecture, Austrin proved that the optimized LLZ algorithm achieves an approximation ratio of $\approx 0.94016567$ for \MAXSAT{2}, where, somewhat surprisingly, the optimizing $\beta$ is also equal to this ratio. 

The LLZ algorithm uses an odd threshold function, i.e., $f(b) = -f(-b)$. However, for $\MAXCSPP$ in general we are allowed to use any $f$ in $\THRESH^-$. We can also use a distribution over $\THRESH^-$ schemes, and such algorithms are called $\THRESH$ by~\cite{LLZ02}.

\subsection{UG-hardness via PCP}\label{sec:PCP}

For 2-CSPs, Austrin showed the following general result that turns hardness against $\THRESH^-$ family into UG-hardness.

\begin{theorem}[\cite{Austrin10}]\label{thm:PCP}
Let \MAXCSPP$(\mathcal{P})$ be a CSP problem where every predicate in $\mathcal{P}$ has arity at most 2. Let $\Theta$ be a distribution of configurations for \MAXCSPP$(\mathcal{P})$ in which every $2$-configuration is positive. Let $c$ be the SDP value of $\Theta$. If no $\THRESH^-$ rounding scheme can satisfy more than an $s$ fraction of configurations in $\Theta$, then it is UG-hard to approximate \MAXCSPP$(\mathcal{P})$ within a ratio of $s/c + \epsilon$ for any $\eps > 0$. 
\end{theorem}

We remark that in the original statement in~\cite{Austrin10}, the CSPs allow arbitrary variable negations, but the proof can be easily extended to the case where negations are not allowed as in our Definition~\ref{def:csp+}, see the appendix of~\cite{BHPZ22} for more details on this. Another small difference is that the original statement does not include $1$-configurations, but this can also be handled easily.

The hard distributions that we construct in subsequent sections will have the property that every $2$-configuration is positive.

\subsection{Interval Arithmetic}

To rigorously analyze these complex rounding probabilities, researchers in approximation algorithms have often used \emph{interval arithmetic} \cite{zwick02,Sjogren09,AuBeGe16,BhangaleKKST18,BhKh20,BHPZ22}. As the name suggests, interval arithmetic keeps track of intervals $[a, b] \subseteq \RR$. When applying a function $f : \RR \to \RR$ to an interval $[a,b]$, the guarantee is that the output $[c,d]$ has that $f(x) \in [c,d]$ for all $x \in [a,b]$. A similar property applies for multivariate functions. Thus, if one seeks to show that a function $g : [0,1]^n \to \RR$ is always nonnegative, it suffices to partition $[0,1]^n$ into finitely many boxes such that the interval arithmetic evaluation of each box is nonnegative.

However, there are many cases in which such a ``divide and conquer'' approach is either impractically slow or literally infeasible. For instance if the function $g$ mentioned actually equals $0$ at some point in the domain, interval arithmetic may never succeed due to rounding errors which introduce small negative numbers. For us to get around that, we also often implement an interval arithmetic implementation of the gradient $\nabla g$ or sometimes even the Hessian of $g$, which lets us prove that the local minima of $g$ are in restricted regions. 

We build our interval arithmetic implementations off of the library Arb~\cite{johansson2017arb} due to it having an efficient implementation of hypergeometric functions (such as the error function) \cite{johansson2019computing} as well as support for integration (to compute 2-dimensional Gaussian cumulative density functions) \cite{johansson2018numerical}. In this paper, any lemma proved using interval arithmetic will be clearly marked ``(interval arithmetic)'' in the theorem statement. A detailed discussion of how these theorems are verified in interval arithmetic is contained in Appendix~\ref{app:IA}.

\section{The worst configurations for \texorpdfstring{\MAXSAT{2}}{MAX 2-SAT} are simple}\label{section:2SAT}

Define $f_\beta(b_1,b_2,b_{12}) = \Prob_\beta(b_1,b_2,b_{12}) - \beta \cdot \Value(b_1,b_2,b_{12})$. The following result is due to Austrin.

\begin{theorem} [\cite{austrin2008conditional}, Proposition 6.6.1]\label{thm:Austrin07}
    There exists $\beta^* = \beta_{LLZ}^- \approx 0.94016567$ such that \[\min_{b \in [-1, 1]} f_{\beta^*}(b, b, -1 + 2|b|) = 0.\]
\end{theorem}

In this section we prove the following theorem which extends the domain of the above minimum to all feasible configurations.

\begin{theorem}\label{thm:Austrin} Let $\beta^*$ be as in Theorem~\ref{thm:Austrin07}. We have $\min_{(b_i, b_j, b_{ij})}f_{\beta^*}(b_i, b_j,\allowbreak b_{ij}) = 0$, where $(b_i, b_j, b_{ij})$ ranges over all feasible configurations.
\end{theorem}

Note that Theorem~\ref{thm:Austrin} immediately implies Conjecture~\ref{conj:Austrin} and $\beta_{LLZ} = \beta_{LLZ}^-$.

Numerical experiments indeed suggest that for the optimal choice $\beta^*\approx 0.94016567$ of~$\beta$, the global minimum is attained simultaneously at the two points $(-b^*,-b^*,-1+2b^*)$ and $(b^*,b^*,-1+2b^*)$, where $b^*\approx 0.162478$. We will show this rigorously. Part of the difficulty, of course, is that we have no closed form expressions for $\beta^*$ and $b^*$. Our proof is composed of the following steps.

\begin{enumerate}
    \item We first show that any minimizer of $f_{\beta^*}(b_1,b_2,b_{12})$ is of the form $(b_1, b_2, -1 + |b_1 + b_2|)$, which means that one of the triangle inequalities in the SDP is tight. This confirms the intuition that for \MAXSAT{2} the triangle inequalities are really cutting away space of very bad configurations. 
    \item We then show that, any point that is far away from $(-b_0,-b_0,-1+2b_0)$ or $(b_0, b_0,-1+2b_0)$ in $\ell_\infty$ distance cannot be a minimizer of  $f_{\beta^*}(b_1,b_2,b_{12})$. Here, $b_0=0.16247834$ is an approximate proxy to $b^*$, as the value of $b^*$ is not exactly known.
    \item Finally, we show that for any bias $b$ near $b_0$ or $-b_0$, the function $g_{b, \beta^*}(t) = \Prob_{\beta^*}(b + t, b - t, -1+2|b|)$ achieves its minimum at $t = 0$ in a large enough neighborhood of $0$. Since $\Value(b + t, b - t, -1+2|b|) = \Value(b, b, -1+2|b|)$, this implies that near $b_0$ or $-b_0$, $f_{\beta^*}(b_1, b_2, -1 + |b_1 + b_2|)$ is minimized at a point where $b_1 = b_2$.
\end{enumerate}

The first two steps will be proven using the technique of interval arithmetic, while the third step will be proven analytically.

\subsection{Step 1}

We prove the following statement using interval arithmetic. The implementation details of each lemma are given in Appendix~\ref{app:IA}.

\begin{lemma}[Interval Arithmetic]\label{lem:new_step1}
    For every $b_1, b_2, b_{12} \in [-1,1]$ and $\beta \in [0.9401653,9401658]$, at least one of the following is true:
    \begin{itemize}
        \item $f_\beta(b_1, b_2, b_{12}) > 0.001$, (cannot be global minimum)
        \item $\rho(b_1, b_2, b_{12}) \notin [-1, 1]$, (triangle inequality violation)
        \item $b_{12} < -1 + |b_1 + b_2|$, (triangle inequality violation)
        \item $b_{12} < 1 - |b_1 - b_2|$ and $\nabla f\neq (0, 0, 0)$. (global minimum cannot be in interior)
    \end{itemize}
\end{lemma}

We remark that the above lemma implies that any minimizer of $f_{\beta^*}$ must satisfy $b_{12} = -1 + |b_1 + b_2|$, since for every critical point $(b_1,b_2,b_{12})$ within the feasible region we will have $f_{\beta^*}(b_1, b_2, b_{12}) > 0.001$.

\subsection{Step 2}
Observe that $f_\beta$ has the following property.
\begin{proposition}\label{prop:3.3}
    For every $\beta \in [-1, 1]$ and every feasible configuration $(b_1, b_2, b_{12})$, we have $f_\beta(b_1, b_2,  b_{12}) = f_\beta(-b_1, -b_2,  b_{12})$.
\end{proposition}
\begin{proof}
    We have
    \begin{align*}
    & \,\,\,\;\;\quad \Value(b_1, b_2, b_{12}) - \Value(-b_1, -b_2, b_{12}) \\
    & \EQ \frac{3 - b_1 - b_2 - b_{12}}{4} - \frac{3 + b_1 + b_2 - b_{12}}{4} \\
    & \EQ -\frac{b_1 + b_2}{2}\;.
    \end{align*}
    We also have $\rho(b_1, b_2,  b_{12}) = \rho(-b_1, -b_2,  b_{12})$ and
    \begin{align*}
        & \,\,\,\;\;\quad \Prob_\beta(b_1, b_2,  b_{12}) - \Prob_\beta(-b_1, -b_2,  b_{12}) \\
    & \EQ \left(1 - \Phi_{\rho}\left(\Phi^{-1}\left(\frac{1 + \beta b_1}{2}\right), \Phi^{-1}\left(\frac{1 + \beta b_2}{2}\right)\right)\right) - \left(1 - \Phi_{\rho}\left(\Phi^{-1}\left(\frac{1 - \beta b_1}{2}\right), \Phi^{-1}\left(\frac{1 - \beta b_2}{2}\right)\right)\right) \\
    & \EQ \Phi_{\rho}\left(\Phi^{-1}\left(\frac{1 - \beta b_1}{2}\right), \Phi^{-1}\left(\frac{1 - \beta b_2}{2}\right)\right) - \Phi_{\rho}\left(\Phi^{-1}\left(\frac{1 + \beta b_1}{2}\right), \Phi^{-1}\left(\frac{1 + \beta b_2}{2}\right)\right) \\
    & \EQ -\beta \cdot \frac{b_1 + b_2}{2}\;.
    \end{align*}
    Here in the last step we have used Proposition~\ref{prop:bivariate_symmetry}. It follows that $f_\beta(b_1, b_2,  b_{12}) - f_\beta(-b_1, -b_2,  b_{12}) = (\Prob_\beta(b_1, b_2,  b_{12}) - \Prob_\beta(-b_1, -b_2,  b_{12})) - \beta (\Value(b_1, b_2, b_{12}) - \Value(-b_1, -b_2, b_{12})) = 0$. 
\end{proof}

The above proposition allows us to assume without loss of generality that $b_1 + b_2 \geq 0$ in the following lemma.

\begin{lemma}[Interval Arithmetic]\label{lem:new_step2}
    For every $b_1, b_2 \in [-1, 1]$ with $b_1 + b_2 \geq 0$ and $\beta \in [0.9401653,0.9401658]$, at least one of the following is true:
    \begin{itemize}
        \item $f_\beta(b_1, b_2, -1 + b_1 + b_2) > 0.001$,
        \item $b_1 + b_2 >0$ and $\left(\frac{\partial}{\partial b_1}f_\beta(b_1, b_2, -1 + b_1 + b_2), \frac{\partial}{\partial b_2}f_\beta(b_1, b_2, -1 + b_1 + b_2)\right) \neq (0, 0)$.
        \item $b_1, b_2 \in [b_0 - \eps, b_0 + \eps]$ where $b_0 = 0.16247834$ and $\eps = 10^{-6}$. 
    \end{itemize}
    Further, we have that
    \begin{itemize}
    \item $f_{0.9401658}(b_0, b_0, -1 + 2b_0) < 0$
    \item For all $b_1, b_2 \in [b_0 - \eps, b_0 + \eps]$, we have that $f_{0.9401653}(b_1, b_2, -1 + b_1 + b_2) > 0$
    \end{itemize}
\end{lemma}

We remark that any minimizer of $f_{\beta^*}$ must satisfy the third condition, since any configuration satisfying the first condition has $f_{\beta^*} \geq 0.001$ and for any configuration satisfying the second condition we can move along a non-zero gradient direction to obtain an even smaller $f_{\beta^*}$.

\subsection{Step 3}

Recall that we defined 
    \[g_{b, \beta}(t)\EQ \Prob_\beta(b + t, b - t, -1+2|b|) \EQ 1 - \Phi_{\rho}\left(t_1, t_2\right)\;,\]
where
\begin{align*}
\rho & \EQ \rho_{b}(t) \EQ \frac{-1 + 2|b| - (b+t)(b-t)}{\sqrt{(1 - (b+t)^2)(1 - (b - t)^2)}}\;, \\
t_1 & \EQ t_{1, b, \beta}(t) \EQ \Phi^{-1}\left(\frac{1 + \beta(b + t)}{2}\right)\;,\\
t_2 & = t_{2, b, \beta}(t) \EQ \Phi^{-1}\left(\frac{1 + \beta(b - t)}{2}\right)\;.
\end{align*}

We show the following proposition:

\begin{proposition}\label{prop:new_step3}
     For $\beta \geq 0.94$, $t \in [-0.1, 0.1]$, $b \in [-0.18, -0.14] \cup [0.14, 0.18]$, we have $g_{b, \beta}'(t) \geq 0$ when $t > 0$ and $g_{b, \beta}'(t) \leq 0$ when $t < 0$.
\end{proposition}

\begin{proof}
Fix some $\beta \ge 0.94$. Since $g_{b,\beta}(t)$ is an odd function in $t$, it suffices to show that for $t \in [0, 0.1]$, $g_{b, \beta}'(t) \geq 0$. Proposition~\ref{prop:bivariate_symmetry} implies that
\begin{align*}
 \Phi_{\rho}(t_1, t_2) & \EQ \Phi_{\rho}(-t_1, -t_2) + \Phi(t_1) + \Phi(t_2) - 1 \\
 & \EQ \Phi_{\rho}(-t_1, -t_2) + \frac{1 + \beta(b + t)}{2} + \frac{1 + \beta(b - t)}{2} - 1 \\
 & \EQ \Phi_{\rho}(-t_1, -t_2) + \beta b\;. 
\end{align*}

Since we also have the symmetries $\rho_b(t) = \rho_{-b}(t)$ and
\[
-t_{1, b, \beta}(t) \EQ -\Phi^{-1}\left(\frac{1 + \beta(b + t)}{2}\right) \EQ \Phi^{-1}\left(\frac{1 + \beta(-b - t)}{2}\right) \EQ t_{2, -b, \beta}(t)\;,
\]
we can deduce that
\begin{align*}
    g_{b, \beta}(t) - g_{-b, \beta}(t) & \EQ \Prob_\beta(b + t, b - t, -1+2|b|) - \Prob_\beta(- b + t, - b - t, -1+2|b|) \\
    & \EQ \left(1 - \Phi_{\rho_b}\left(t_{1, b, \beta}, t_{2, b, \beta}\right)\right) - \left(1 - \Phi_{\rho_{-b}}\left(t_{1, -b, \beta}, t_{2, -b, \beta}\right)\right) \\
    & \EQ \left(1 - \Phi_{\rho_b}\left(t_{1, b, \beta}, t_{2, b, \beta}\right)\right) - \left(1 - \Phi_{\rho_{b}}\left(-t_{1, b, \beta}, -t_{2, b, \beta}\right)\right) \\
    & \EQ -\beta b\;.
\end{align*}
This shows that by flipping $b$ to $-b$ we incur a constant change in $g_{b, \beta}(t)$, so to the derivative it is sufficient to fix some $b \in [0.14, 0.18]$, the positive part of the interval that we are interested in. We will also drop $b, \beta$ in $g_{b, \beta}$ and refer to it simply as $g$.

To compute $g'(t)$, we first compute the first derivatives of $t_1, t_2$ with respect to $t$. By the inverse function rule and the chain rule, we have
\[
\frac{\partial t_1}{\partial t} \EQ \frac{\beta}{2}\cdot \frac{1}{\varphi(t_1)}\quad, \quad \frac{\partial t_2}{\partial t} \EQ -\frac{\beta}{2}\cdot \frac{1}{\varphi(t_2)}\;.
\]
Now, for $g'(t)$,  we have
\begin{align*}
    g'(t) & \EQ - \left( \frac{\partial \Phi_\rho}{\partial \rho}(t_1, t_2) \cdot \frac{\partial \rho}{\partial t} + \frac{\partial \Phi_\rho}{\partial x}(t_1, t_2) \cdot \frac{\partial t_1}{\partial t} + \frac{\partial \Phi_\rho}{\partial y}(t_1, t_2) \cdot \frac{\partial t_2}{\partial t} \right)\;.
\end{align*}
Using Lemma~\ref{L:phi_rho}, we have
\[
\frac{\partial \Phi_\rho}{\partial x}(t_1, t_2) \EQ \varphi(t_1)\Phi\left(\frac{t_2 - \rho t_1}{\sqrt{1 - \rho^2}}\right) \quad, \quad \frac{\partial \Phi_\rho}{\partial y}(t_1, t_2) \EQ \varphi(t_2)\Phi\left(\frac{t_1 - \rho t_2}{\sqrt{1 - \rho^2}}\right).
\]
Since $\varphi(x) = \varphi(-x)$, this implies that
\begin{align*}
    g'(t) & \EQ - \left( \frac{\partial \Phi_\rho}{\partial \rho}(t_1, t_2) \cdot \frac{\partial \rho}{\partial t} + \frac{\beta}{2} \cdot \left(\Phi\left(\frac{t_2 - \rho t_1}{\sqrt{1 - \rho^2}}\right) - \Phi\left(\frac{t_1 - \rho t_2}{\sqrt{1 - \rho^2}}\right)\right) \right)\;.
\end{align*}
In order to prove $g'(t) \geq 0$, it is sufficient to show that
\[
 \frac{\beta}{2} \cdot \left(\Phi\left(\frac{t_1 - \rho t_2}{\sqrt{1 - \rho^2}}\right) - \Phi\left(\frac{t_2 - \rho t_1}{\sqrt{1 - \rho^2}}\right)\right) \GE \frac{\partial \Phi_\rho}{\partial \rho}(t_1, t_2) \cdot \frac{\partial \rho}{\partial t}\;,
\]
which, by Proposition~\ref{prop:Phi_diff}, follows from
\[
\frac{\beta}{2} \cdot \frac{(t_1 - \rho t_2) - (t_2 - \rho t_1)}{\sqrt{1 - \rho^2}} \cdot \min\left(\varphi\left(\frac{t_1 - \rho t_2}{\sqrt{1 - \rho^2}}\right), \varphi\left(\frac{t_2 - \rho t_1}{\sqrt{1 - \rho^2}}\right)\right) \GE \frac{\partial \Phi_\rho}{\partial \rho}(t_1, t_2) \cdot \frac{\partial \rho}{\partial t} \tag{$\ast$}\;.
\]
We have
\[
\frac{(t_1 - \rho t_2) - (t_2 - \rho t_1)}{\sqrt{1 - \rho^2}} \EQ \frac{(t_1 - t_2)(1 + \rho)}{\sqrt{1 - \rho^2}} \GE 0\;,
\]
since $t_1 \geq t_2$. To prove $(\ast)$, we can prove separately that
\[
\frac{\beta}{2} \cdot \frac{(t_1 - t_2)(1 + \rho)}{\sqrt{1 - \rho^2}} \cdot \varphi\left(\frac{t_1 - \rho t_2}{\sqrt{1 - \rho^2}}\right) \GE \frac{\partial \Phi_\rho}{\partial \rho}(t_1, t_2) \cdot \frac{\partial \rho}{\partial t} \;,
\]
and
\[
\frac{\beta}{2} \cdot \frac{(t_1 - t_2)(1 + \rho)}{\sqrt{1 - \rho^2}} \cdot \varphi\left(\frac{t_2 - \rho t_1}{\sqrt{1 - \rho^2}}\right) \GE \frac{\partial \Phi_\rho}{\partial \rho}(t_1, t_2) \cdot \frac{\partial \rho}{\partial t} \;.
\]

For the former, we can rewrite it by plugging in expressions for $\varphi$ and $\frac{\partial \Phi_\rho}{\partial \rho}$ (Lemma~\ref{L:phi_rho}) and get
\[
\frac{\beta}{2} \cdot \frac{(t_1 - t_2)(1 + \rho)}{\sqrt{1 - \rho^2}} \cdot \frac{1}{\sqrt{2\pi}}\exp\left(-\frac{(t_1 - \rho t_2)^2}{2(1 - \rho^2)}\right) \GE \frac{1}{2\pi \sqrt{1 - \rho^2}} \exp\left(-\frac{t_1^2 - 2\rho t_1t_2 + t_2^2}{2(1 - \rho^2)}\right) \cdot \frac{\partial \rho}{\partial t} \;,
\]
which can then be further simplified to
\[
\frac{\beta}{2} \cdot (t_1 - t_2)(1 + \rho) \GE \varphi(t_2) \cdot \frac{\partial \rho}{\partial t} \;.
\]
Similarly for the latter we can simplify it to
\[
\frac{\beta}{2} \cdot (t_1 - t_2)(1 + \rho) \GE \varphi(t_1) \cdot \frac{\partial \rho}{\partial t} \;.
\]
Therefore, if we show that 
\[
\frac{\beta}{2} \cdot (t_1 - t_2)(1 + \rho) \GE \max(\varphi(t_1), \varphi(t_2))\cdot \frac{\partial \rho}{\partial t}\;, \tag{1}
\]
then it will follow that $g'(t) \geq 0$.

We have
\begin{align*}
t_1 - t_2 & \EQ \Phi^{-1}\left(\frac{1 + \beta(b + t)}{2}\right) - \Phi^{-1}\left(\frac{1 + \beta(b - t)}{2}\right) \\
& \GE \left(\frac{1 + \beta(b + t)}{2} - \frac{1 + \beta(b - t)}{2}\right) \cdot \min_{t \in \mathbb{R}} (\Phi^{-1})'(t) \\
& \EQ \left(\frac{1 + \beta(b + t)}{2} - \frac{1 + \beta(b - t)}{2}\right) \cdot \min_{t \in \mathbb{R}} \frac{1}{\varphi(t)} \\
& \EQ \beta t \cdot \sqrt{2\pi}\;.
\end{align*}

We also have 
\[
\max(\varphi(t_1), \varphi(t_2)) \LE \frac{1}{\sqrt{2\pi}}\;.
\]

Plugging all these into (1), as a very loose estimate, if we can show that
\[
\beta^2\pi(1 + \rho) \cdot t \GE \frac{\partial \rho}{\partial t}\;.
\]
then the proof is complete. This inequality follows from the estimates in Propositions~\ref{prop:RHS} and~\ref{prop:LHS}. 
\end{proof}

\begin{remark}\label{remark:beta_to_gamma}
    We remark that the above estimation would still work if we replace $t_1$ and $t_2$ with $t_1 = \Phi^{-1}\left(\frac{\beta + \beta(b + t)}{2}\right)$ and $t_2 = \Phi^{-1}\left(\frac{\beta + \beta(b - t)}{2}\right)$, since we would still have $\frac{\partial t_1}{\partial t} \EQ \frac{\beta}{2}\cdot \frac{1}{\varphi(t_1)},  \frac{\partial t_2}{\partial t} \EQ -\frac{\beta}{2}\cdot \frac{1}{\varphi(t_2)}$, which is the only role played by the expression of $t_1$ and $t_2$.
\end{remark}

\begin{proposition}\label{prop:RHS}
    For $t \in [0, 0.1]$, $b \in [0.14, 0.19]$, we have $    0 \leq \frac{\partial \rho}{\partial t} <\frac{2t}{3} $.
\end{proposition}
\begin{proof}
    We can write $\rho$ as
\[
\rho \EQ \frac{-1 + 2|b| - (b+t)(b-t)}{\sqrt{(1 - (b+t)^2)(1 - (b - t)^2)}} \EQ \frac{-(1 - |b|)^2 + t^2}{\sqrt{(1 - (b^2 + t^2))^2 - 4b^2t^2}} \;.
\]
We have
\begin{align*}
\frac{\partial \rho}{\partial t} & \EQ 2t \cdot \frac{((1 - (b^2 + t^2))^2 - 4b^2t^2) + (1 + (b^2 - t^2))\cdot(- (1-b)^2 + t^2))}{((1 - (b^2 + t^2))^2 - 4b^2t^2)^{3/2}} \\
& \EQ t \cdot \frac{4b((1-b)^2 - t^2)}{((1 - (b^2 + t^2))^2 - 4b^2t^2)^{3/2}} \;.
\end{align*}

    It suffices to show that for parameters in these range we have
    \[
    0 \LE \frac{4b((1-b)^2 - t^2)}{((1 - (b^2 + t^2))^2 - 4b^2t^2)^{3/2}} \LE \frac{2}{3} \;.
    \]
    For the denominator we have
    \[
    ((1 - (b^2 + t^2))^2 - 4b^2t^2)^{3/2} \GE ((1 - (0.19^2 + 0.1^2))^2 - 4 \cdot 0.19^2 \cdot 0.1^2)^{3/2} \GE 0.8659 \;.
    \]
    For the numerator we have
    \[
    0 \leq 4b((1-b)^2 - t^2) \LE 4b(1-b)^2 \LE 4 \times 0.19 (1 - 0.19)^2 < 0.499 \;.
    \]

    Therefore the value of the fraction is bounded between $0$ and $0.499/0.8659 < 2/3$.
\end{proof}

\begin{proposition}\label{prop:LHS}
    For $\beta \geq 0.94$, $t \in [0, 0.1]$, $b \in [0.14, 0.19]$, we have $
    \beta^2\pi(1 + \rho(t)) \geq 0.681$.
\end{proposition}
\begin{proof}
    We have
    \begin{align*}
    1 + \rho(t) \GE 1 + \rho(0) \EQ 1 + \frac{-1 + 2b - b^2}{1 - b^2} \EQ 1 + \frac{-(1-b)}{(1+b)} = \EQ 2 - \frac{2}{1 + b}\;.
    \end{align*}
    It follows that
    \[
    \beta^2\pi(1 + \rho(t)) \GE 0.94^2 \cdot \pi \cdot \left(2 - \frac{2}{1 + 0.14}\right) \GT 0.681\;. \qedhere
    \]
\end{proof}

\subsection{Wrapping-Up}
\begin{proof}[Proof of Theorem~\ref{thm:Austrin}]
     Let $(b_1, b_2, b_{12})$ be a feasible configuration such that $f_{\beta^*}(b_1, b_2, b_{12})$ is minimized. By Lemma~\ref{lem:new_step1}, we have that $b_{12} = -1 + |b_1 + b_2|$. By Proposition~\ref{prop:3.3} and Lemma~\ref{lem:new_step2}, we have that furthermore $b_1, b_2 \in [b_0 - \eps, b_0 + \eps]$ or $b_1, b_2 \in [-b_0 - \eps,- b_0 + \eps]$ where $b_0 = 0.16247834$ and $\eps = 10^{-6}$. Finally, by Proposition~\ref{prop:new_step3}, for $b_1, b_2$ in this range, we have

    \begin{align*}
    f_{\beta^*}(b_1, b_2, -1 + |b_1 + b_2|) & \GE f_{\beta^*}\left(\frac{b_1 + b_2}{2}, \frac{b_1 + b_2}{2},  -1 + |b_1 + b_2|\right).
    \end{align*}
    In other words, we have found a simple configuration $(b, b, -1 + 2|b|)$ such that 
    \[
    f_{\beta^*}(b_1, b_2, b_{12}) \GE f_{\beta^*}(b, b, -1 + 2|b|)\;.
    \]
    This completes our proof.
\end{proof}

\section{MAX CSP$(\{x\lor y,x,\bar{x}\})$}\label{sec:ORNOT}
In this section we consider MAX CSP$(\{x\lor y,x,\bar{x}\})$, i.e., the version of the \MAXSAT{2} problem in which all literals appearing in clauses of size 2 are positive. By the classification of Khanna, Sudan, Trevisan and Williamson \cite{KSTW01}, this problem is still APX-hard. As mentioned, it can be viewed as a variant of the \VC\ problem.

\subsection{The Rounding Scheme}
We consider the $\THRESH^-$ rounding scheme $b \mapsto -1 + \gamma(1 + b)$ for some parameter $\gamma \in [0, 1]$. For any $\bar{x}$ constraint with bias $b$, we have that its SDP value is equal to $\frac{1+x}{2}$, while $b \mapsto -1 + \gamma(1 + b)$ satisfies it with probability $\Phi\left(\Phi^{-1}\left(\frac{1 - 1 + \gamma(1 + b)}{2}\right)\right) = \frac{\gamma(1 + b)}{2}$. For any $x$ constraint with bias $b$, its SDP value is $\frac{1 - x}{2}$ and $b \mapsto -1 + \gamma(1 + b)$ satisfies it with probability 
\[
1 - \Phi\left(\Phi^{-1}\left(\frac{1 - 1 + \gamma(1 + b)}{2}\right)\right) = 1 -  \frac{\gamma(1 + b)}{2} \geq \frac{\gamma(1-b)}{2}.
\]
This shows that the parameter $\gamma$ is also the approximation ratio that the rounding scheme achieves on unary constraints. Note that the function $-1 + \gamma(1 + b)$ is increasing in $\gamma$ for every $b$, which means we are less likely to satisfy any $x \vee y$ constraints if we increase $\gamma$. To optimize $\gamma$, it is then sufficient to find $\gamma = \gamma^*$ such that $b \mapsto -1 + \gamma^*(1 + b)$ achieves an approximation ratio of also $\gamma^*$ on $2$-configurations. 

We will use a similar approach from the previous section. Similar to $f_\beta$, we define $h_\gamma(b_1, b_2, b_{12}) = \left(1 - \Phi_\rho\left(\Phi^{-1}\left(\frac{\gamma(1 + b_1)}{2}\right), \Phi^{-1}\left(\frac{\gamma(1 + b_2)}{2}\right)\right)\right) - \gamma \Value(b_1, b_2, b_{12})$ where $\rho = \rho(b_1, b_2, b_{12})$.

\begin{proposition}\label{prop:type3_monotone}
    For every feasible configuration $(b_1, b_2, b_{12})$, $h_\gamma(b_1, b_2,b_{12})$ monotonically decreases with $\gamma$. In particular $\min_{(b_1, b_2, b_{12})} h_\gamma(b_1, b_2,b_{12})$ decreases with $\gamma$, where the minimum is taken over all feasible configurations. Furthermore $\min_{(b_1, b_2, b_{12})} h_{\gamma^*}(b_1, b_2,b_{12}) = 0$.
\end{proposition}
\begin{proof}
    The proposition holds since $\Phi_\rho\left(\Phi^{-1}\left(\frac{\gamma(1 + b_1)}{2}\right), \Phi^{-1}\left(\frac{\gamma(1 + b_2)}{2}\right)\right)$ and $\gamma \Value(b_1, b_2, b_{12})$ are both increasing in $\gamma$. Note that $\min h_{\gamma}(b_1, b_2,b_{12}) \geq 0$ implies that $-1 + \gamma(1 + b)$ achieves an approximation ratio of at least $\gamma$, and therefore for the optimal $\gamma = \gamma*$ the equality must be achieved.
\end{proof}

We use interval arithmetic to certify the following lemmas.
\begin{lemma}[Interval arithmetic]\label{lem:type3-step1}
For every $b_1, b_2, b_{12} \in [-1,1]$ and $\gamma \in [0.95, 0.96]$, at least one of the following is true:
\begin{itemize}
    \item $h_\gamma(b_1, b_2, b_{12}) > 0.001$, (cannot be global minimum)
    \item $\rho(b_1, b_2, b_{12}) \notin [-1, 1]$, (triangle inequality violation)
    \item $b_{12} < -1 + |b_1 + b_2|$, (triangle inequality violation)
    \item $b_{12} < 1 - |b_1 - b_2|$ and $\nabla h_\gamma \neq (0, 0, 0)$ (global minimum cannot be in interior)
\end{itemize}
\end{lemma}

\begin{lemma}[Interval arithmetic]\label{lem:type3-step2}
For every $b_1, b_2 \in [-1, 1]$ and $\gamma \in [0.9539798, 0.95398]$ at least one of the following is true:
\begin{itemize}
    \item $h_\gamma(b_1, b_2, -1 + |b_1 + b_2|) > 0.001$,
    \item $b_1 + b_2 < 0$ and $(\frac{\partial}{\partial b_1}h_\gamma(b_1, b_2, -1 - b_1 - b_2), \frac{\partial}{\partial b_2}h_\gamma(b_1, b_2, -1 - b_1 - b_2)) \neq (0, 0)$,
    \item $b_1, b_2 \in [b_0 - \eps, b_0 + \eps]$ where $b_0 = -0.1824167935$ and $\eps = 10^{-6}$.
\end{itemize}
Further, we have
\begin{itemize}
\item $h_{0.95398}(b_0, b_0, -1 - 2b_0) < 0$
\item For all $b_1, b_2 \in [b_0 - \eps, b_0 + \eps]$, $h_{0.9539798}(b_1, b_2, -1 - b_1 - b_2) > 0$.
\end{itemize}
\end{lemma}

\begin{theorem}\label{thm:type3}
    We have $\gamma^* \in [0.9539798, 0.95398]$. Furthermore, the minimum in the expression $\min_{\theta = (b_1, b_2, b_{12})} h_{\gamma^*}(b_1, b_2, b_{12})$ is achieved at some point $(b, b, -1 - 2b)$ for some $b = \in [b_0 - \eps, b_0 + \eps]$ where $b_0 = -0.1824167935$ and $\eps = 10^{-6}$.
\end{theorem}

\begin{proof}
    By Lemma~\ref{lem:type3-step1} and Lemma~\ref{lem:type3-step2}, any local minimum of $h_\gamma$ in the feasible region has either value $>0.001$ or has biases $b_1, b_2 \in [b_0 - \eps, b_0 + \eps]$ and $b_{12} = -1 - b_1 - b_2$ where $b_0 = -0.1824167935$ and $\eps = 10^{-6}$. Since for all $b_1, b_2 \in [b_0 - \eps, b_0 + \eps]$, $h_{0.9539798}(b_1, b_2, -1 -b_1 - b_2) > 0$, we can deduce that $\min h_{0.9539798}\geq 0$, while on the other hand $\min h_{0.95398} \leq h_{0.95398}(b_0, b_0, -1 - 2b_0) < 0$. By Proposition~\ref{prop:type3_monotone}, this implies that $\gamma^* \in [0.9539798, 0.95398]$.

    For the second claim, note that since $\min h_{\gamma^*} = 0$, using the previous logic, the minimum in the feasible region must be achieved with biases $b_1, b_2 \in [b_0 - \eps, b_0 + \eps]$ and $b_{12} = -1 - b_1 - b_2$. By Proposition~\ref{prop:new_step3} and Remark~\ref{remark:beta_to_gamma}, we have 
    \[
    h_{\gamma^*}(b_1, b_2, -1 - b_1 - b_2) \geq h_{\gamma^*}\left(\frac{b_1 + b_2}{2}, \frac{b_1 + b_2}{2}, -1 - b_1 - b_2\right).
    \]
    This completes our proof.
\end{proof}

\subsection{Matching Hardness}\label{sec:type3-hardness}

Let $b^\ast = b(\gamma^*) \approx -0.1824$ be the hardest bias for $\gamma^*$ from Theorem~\ref{thm:type3} and $p_1, p_2 \in [0, 1]$ be some parameters to be chosen later. Consider the following distribution $\Theta_1$ of configurations.

\begin{center}
\begin{tabular}{|l|c|l|}
\hline
    Configuration & Probability & Predicate type \\ \hline
    $\theta_1 \EQ (b^\ast, b^\ast, -1 - 2b^\ast)$ &  $p_1$ & $x \lor y$ \\ \hline
    $\theta_2 \EQ (b^\ast)$ &  $p_2$ & $\bar{x}$ \\ \hline
\end{tabular}
\end{center}

We will prove hardness against all $\THRESH^-$ rounding schemes $b \mapsto f(b)$. Since there is only one bias involved, it is sufficient to consider the threshold for that bias. Let $\rho = \rho(b^\ast) = - \frac{1 - b^\ast}{1 + b^\ast}$. Let $\Value(\Theta_1)$ be the SDP value of this distribution and $\Prob(\Theta_1, t)$ be the probability of satisfying a configuration sampled from $\Theta_1$ if $f(b^\ast) = 2\Phi(t)-1$ (We choose this parametrization so that $f$ sets any variable with bias $b$ to true with probability $(1 - \Phi(t))$, for convenience of the analysis).

\begin{proposition}
    We have $\Value(\Theta_1) = p_1 + p_2 \cdot \frac{1 - b^*}{2}$ and $\Prob(\Theta_1, t) = p_1 \cdot (1 - \Phi_\rho(t, t)) + p_2 \cdot \Phi(t)$. 
\end{proposition}

It is straightforward to find the best threshold $t$ using calculus. We have
\begin{proposition}\label{prop:type3-hard}
    Let $t^\ast = \sqrt{\frac{1 + \rho}{1 - \rho}} \cdot \Phi^{-1}\left(\frac{p_2}{p_1}\right)$. For every $t \in \mathbb{R} \cup \{\pm \infty\}$ we have
    \[
    \Prob(\Theta_1, t) \LE \Prob(\Theta_1, t^\ast)\;.
    \]
\end{proposition}
\begin{proof}
    Using Lemma~\ref{L:phi_rho}, we have
    \begin{align*}
    \frac{\partial}{\partial t} \Prob(\Theta_1, t) & \EQ -p_1 \cdot \varphi(t) \cdot \Phi\left(\sqrt{\frac{1 - \rho}{1 + \rho}} \cdot t\right) + p_2 \cdot \varphi(t) \\
    & \EQ \varphi(t) \cdot \left(-p_1 \cdot \Phi\left(\sqrt{\frac{1 - \rho}{1 + \rho}} \cdot t\right) + p_2\right) \;. \\
    \end{align*}
    The proposition follows since $\frac{\partial}{\partial t} \Prob(\Theta_1, t) > 0$ when $t \leq t_\ast$ and $\frac{\partial}{\partial t} \Prob(\Theta_1, t) < 0$ when $t \geq t_\ast$.
\end{proof}

We now take $\frac{p_2}{p_1} = \Phi\left(\sqrt{\frac{1-\rho}{1 + \rho}} \cdot \Phi^{-1}\left(\frac{\gamma^*(1 + b^*)}{2}\right) \right)$, so that the value of $t^*$ in Proposition~\ref{prop:type3-hard} will coincide with the value given $b \to -1 + \gamma^*(1 + b)$ at $b = b^*$. This shows that for such $p_1$ and $p_2$, $b \to -1 + \gamma^*(1 + b)$ is the optimal $\THRESH^-$ scheme on $\Theta_1$. Since we already know that $b \to -1 + \gamma^*(1 + b)$ has approximation ratio $\gamma^*$ on both configurations in this distribution, Theorem~\ref{thm:PCP} immediately implies the following:

\begin{theorem}\label{thm:type3_tight}
    For every $\epsilon > 0$, it is UG-hard to approximate MAX CSP$(\{x\lor y,x,\bar{x}\}$ within a ratio of $\gamma^* + \epsilon$.
\end{theorem}

\section{MAX $\{1,2\}$-HORN SAT}\label{sec:horn}

In this section we consider MAX $\{1,2\}$-HORN SAT, where the clauses are of types $x\lor y$ , $\bar{x}\lor y$, $x$ and $\bar{x}$. We will first present and analyze a rounding scheme from the $\THRESH$ family. Then we use the hardest configuration for this rounding scheme to design a distribution of configurations for which this rounding scheme is the optimal rounding scheme, thereby obtaining a matching hardness result.

\subsection{The Rounding Scheme}

We first remark that the threshold function for this problem need not be odd. We consider the following $\THRESH$ rounding scheme $F_\alpha$. 

\begin{center}
\begin{tabular}{ll}
    
    $b \mapsto b$ & \text{with probability} $\alpha$ \\
    
    $b \mapsto -1$ & \text{with probability} $1 - \alpha$ \\
    
\end{tabular}
\end{center}

In other words, with some probability $\alpha$ we round with the odd threshold function $b \mapsto b$, and with the remaining probability $1 - \alpha$ we set every variable to true. We need to find $\alpha$ that maximizes the approximation ratio of $F_\alpha$. We have the following property for the optimal $\alpha$.

\begin{proposition}
    Assume that the approximation ratio of $F_\alpha$ is maximized when $\alpha = \alpha^*$. Then the approximation ratio of $F_{\alpha^*}$ is also equal to $\alpha^*$.
\end{proposition}
\begin{proof}

Observe that any 1-configuration with bias $b$ and predicate type $\bar{x}$ has SDP value $\frac{1 + b}{2}$, while the function $b \mapsto b$ satisfies it with probability $\frac{1 + b}{2}$ as well. This implies that $b \mapsto b$ achieves an approximation ratio of 1 on all $\bar{x}$ constraints. On the other hand, if we set every variable to true, then we never satisfy any $\bar{x}$ constraint. Therefore, on the $1$-configurations, $F_\alpha$ has an overall approximation ratio $\alpha$ for every $\alpha$.

Now, note that by setting every variable to true we satisfy all constraints of the form $x \vee y$, $\bar{x} \vee y$ but satisfy no constraints of the form $\bar{x}$, so by decreasing $\alpha$ we increase the approximation ratio on the 2-configurations. This means that for the optimal $\alpha$, $F_\alpha$ must achieve the same approximation ratio on both 1-configurations and 2-configurations, otherwise we can adjust $\alpha$ to increase the approximation ratio.
\end{proof}

Recall that we defined $\Prob_\beta (b_1, b_2, b_{ij}) \EQ 1 - \Phi_{\rho(\theta)}\left(\Phi^{-1}\left(\frac{1+\beta b_1}{2}\right),\Phi^{-1}\left(\frac{1+\beta b_2}{2}\right)\right)$ and $f_\beta(b_1, b_2, b_{12}) = \Prob_\beta (b_1, b_2, b_{12}) - \beta \cdot \Value(b_1, b_2, b_{12})$. Since in this section we will always have $\beta = 1$, we will omit the subscript in $f_\beta$ and simply refer to it as $f$. The following lemma gives an expression for $\alpha^*$.

\begin{lemma}
    $\alpha^*$ satisfies the following equality: 
    \[
    1 - \frac{1}{\alpha^*} \EQ \min_{\theta = (b_1, b_2, b_{12})} f(b_1, b_2, b_{12}) = \min_{\theta = (b_1, b_2, b_{12})} \Prob_{1}(b_1, b_2, b_{12}) - \Value(b_1, b_2, b_{12}) \;,
    \]
    where $\theta$ ranges over all feasible 2-configurations.
\end{lemma}
\begin{proof}
    The probability that $F_{\alpha^*}$ satisfies any 2-configuration $(b_1, b_2, b_{12})$ is given by
    \[
    1 - \alpha^* + \alpha^* \cdot \Prob_1(b_1, b_2, b_{12}) \;,
    \]
    where $1 - \alpha^*$ is contributed by the function $b \mapsto -1$ and $\alpha^* \cdot \Prob_1(b_1, b_2, b_{12})$ is contributed by $b \mapsto b$. Since $F_{\alpha^*}$ achieves an approximation ratio of $\alpha^*$, we have
    \[
        1 - \alpha^* + \alpha^* \cdot \Prob_1(b_1, b_2, b_{12}) \GE \alpha^* \cdot \Value(b_1, b_2, b_{12}) \;.
    \]
    Rearranging, we obtain that
    \[
         1 - \frac{1}{\alpha^*} \LE \Prob_{1}(b_1, b_2, b_{12}) - \Value(b_1, b_2, b_{12}) \;.
    \]
    Since this is true for every feasible configuration and equality is achieved on some configuration, we obtain that   
    \[
     1 - \frac{1}{\alpha^*} = \min_{\theta = (b_1, b_2, b_{12})} \Prob_{1}(b_1, b_2, b_{12}) - \Value(b_1, b_2, b_{12}). \qedhere
    \]
\end{proof}

To find $\min_{\theta = (b_1, b_2, b_{12})} f(b_1, b_2, b_{12})$, we will again use interval arithmetic, following a similar approach from Section~\ref{section:2SAT}. We first show that the minimum point must be on the boundary of the triangle inequality $b_{12} = -1 + |b_1 + b_2|$.

\begin{lemma}[Interval arithmetic]\label{lem:type4-step1}
For $b_1, b_2, b_{12} \in [-1, 1]$, at least one of the following is true.
\begin{itemize}
    \item $f \ge 1 - \frac{1}{0.95}$, (cannot be hardest point)
    \item $\rho(b_1, b_2, b_{12}) \notin [-1, 1]$, (triangle inequality violation)
    \item $b_{12} < -1 + |b_1 + b_2|$, (triangle inequality violation)
    \item $b_{12} < 1 - |b_1 - b_2|$ and $\nabla f \neq 0$. (global minimum cannot be interior)
\end{itemize}
\end{lemma}

We then show that the minimum point has to be close to $(b, b, -1 + 2|b|)$ or $(-b, -b, -1 + 2|b|)$, where $b = 0.1489442$ is a numerical approximation to the bias in the hardest configuration.

\begin{lemma}[Interval arithmetic]\label{lem:type4-step2}
For $b_1, b_2\in [-1, 1]$ with $b_1 + b_2 \geq 0$ at least one of the following is true:

    \begin{itemize}
        \item $f(b_1, b_2, -1 + b_1 + b_2) >  1 - \frac{1}{0.95}$,
        \item $b_1 + b_2 >0$ and $\left(\frac{\partial}{\partial b_1}f(b_1, b_2, -1 + b_1 + b_2), \frac{\partial}{\partial b_2}f(b_1, b_2, -1 + b_1 + b_2)\right) \neq (0, 0)$.
        \item $b_1, b_2 \in [b_0 - \eps, b_0 + \eps]$ where $b_0 = 0.1489442419$ and $\eps = 10^{-6}$. 
    \end{itemize}

    Further, $f(b_0, b_0, -1 + 2b_0) < 1 - \frac{1}{0.9462}$.
\end{lemma}

\begin{theorem}\label{thm:type4}
    The minimum in the expression $\min_{\theta = (b_1, b_2, b_{12})} f(b_1, b_2, b_{12})$ is achieved at some point $(b^*, b^*, -1 + 2|b^*|)$ for some $b^*$ with $|b^*| \in [b_0 - \eps, b_0 + \eps]$ where $b_0 = 0.1489442$ and $\eps = 10^{-6}$. 
\end{theorem}
\begin{proof}
    Let $(b_1, b_2, b_{12})$ be a feasible configuration on which the minimum is achieved. Lemma~\ref{lem:type4-step1} and Lemma~\ref{lem:type4-step2} shows that we must have $b_{12} = -1 + |b_1 + b_2|$, and either $b_1, b_2 \in [b_0 - \eps, b_0 + \eps]$ or $b_1, b_2 \in [-b_0 - \eps, -b_0 + \eps]$. For biases in this range, by Proposition~\ref{prop:new_step3} (which is applicable since $\beta = 1 > 0.94$), we have 
    \[
    f(b_1, b_2, -1 + |b_1 + b_2|) \geq f\left(\frac{b_1 + b_2}{2}, \frac{b_1 + b_2}{2}, -1 + |b_1 + b_2|\right).
    \]
    This means that the minimum is achieved on $\left(\frac{b_1 + b_2}{2}, \frac{b_1 + b_2}{2}, -1 + |b_1 + b_2|\right)$ as well, and this completes the proof.
\end{proof}

\subsection{Hard Configurations}

Let $b^* > 0$ be a hardest bias on which $f$ achieves its minimum as in Theorem~\ref{thm:type4}. Using $b^*$, we construct the following distribution $\Theta_2$.
\begin{center}
\begin{tabular}{|l|c|l|}
\hline
    Configuration & Probability & Predicate type \\ \hline
    $\theta_1 \EQ (-b^*, -b^*, -1 + 2b^*)$ &  $p_1$ & $x \lor y$ \\ \hline
    $\theta_2 \EQ (b^*, b^*, -1 + 2b^*)$ &  $p_2$ & $x \lor y$ \\ \hline
    $\theta_3 \EQ (-b^*, b^*, 1 - 2b^*)$ &  $p_3$ & $\bar{x} \lor y$ \\ \hline
    $\theta_4 \EQ (b^*, -b^*, 1 - 2b^*)$ &  $p_4$ & $\bar{x} \lor y$ \\ \hline
    $\theta_5 \EQ (-b^*)$ &  $p_5$ & $\bar{x}$ \\ \hline
    $\theta_6 \EQ (b^*)$ &  $p_6$ & $\bar{x}$ \\ \hline   
\end{tabular}
\end{center}

 Let $\rho = \rho(b^*) = - \frac{1-b^*}{1+b^*}$ be the relative pairwise bias of the configuration $(-b^*, -b^*, -1 + 2b^*)$.
 
 Let $\Value(\Theta_2)$ be the SDP value of this distribution and $\Prob(\Theta_2, t_1, t_2)$ be the probability of a $\THRESH^-$ scheme $f$ satisfying a configuration sampled from $\Theta_2$ if $f(-b^*) = 2\Phi(t_1)-1$ and  $f(b^*) = 2\Phi(t_2)-1$, similarly as defined in Section~\ref{sec:type3-hardness}.
 
 \begin{proposition}
     We have
     \[
     \Value(\Theta_2) = (p_1 + p_4) + (p_2 + p_3) \cdot (1 - b^*) + p_5 \cdot \frac{1 - b^*}{2} + p_6 \cdot \frac{1 + b^*}{2}
     \]
     and
     \begin{align*}
     \Prob(\Theta_2, t_1, t_2) = & \,\, p_1 \cdot (1 - \Phi_\rho(t_1, t_1)) + p_2 \cdot (1 - \Phi_\rho(t_2, t_2)) + p_3 \cdot (1 - \Phi_\rho(-t_1, t_2)) \\
     &+ p_4 \cdot (1 - \Phi_\rho(-t_2, t_1)) + p_5 \cdot \Phi(t_2) + p_6 \cdot \Phi(t_1).
     \end{align*}
 \end{proposition}

 We also have the following partial derivatives for $\Prob(\Theta_2, t_1, t_2)$.
 
 \begin{proposition}\label{prop:type4_hard_partials}
     We have
     \begin{align*}
     & \qquad \frac{\partial}{\partial t_1}\Prob(\Theta_2, t_1, t_2) \\ 
     & = -2p_1 \varphi(t_1) \Phi\left(\sqrt{\frac{1 - \rho}{1 + \rho}} t_1\right) + p_3 \varphi(t_1)\Phi\left(\frac{t_2 + \rho t_1}{\sqrt{1 - \rho^2}}\right) - p_4\varphi(t_1)\Phi\left(\frac{-t_2 - \rho t_1}{\sqrt{1 - \rho^2}}\right) + p_6\varphi(t_1) \\  
     & = \varphi(t_1) \cdot \left(-2p_1  \Phi\left(\sqrt{\frac{1 - \rho}{1 + \rho}} t_1\right) + p_3 \Phi\left(\frac{t_2 + \rho t_1}{\sqrt{1 - \rho^2}}\right) - p_4\Phi\left(\frac{-t_2 - \rho t_1}{\sqrt{1 - \rho^2}}\right) + p_6\right)
     \end{align*}
     and
     \begin{align*}
     & \qquad \frac{\partial}{\partial t_2}\Prob(\Theta_2, t_1, t_2) \\ 
     & = -2p_2 \varphi(t_2) \Phi\left(\sqrt{\frac{1 - \rho}{1 + \rho}} t_2\right) - p_3 \varphi(t_2)\Phi\left(\frac{-t_1 - \rho t_2}{\sqrt{1 - \rho^2}}\right) + p_4\varphi(t_2)\Phi\left(\frac{t_1 + \rho t_2}{\sqrt{1 - \rho^2}}\right) + p_5\varphi(t_2) \\  
     & = \varphi(t_2) \cdot \left(-2p_2  \Phi\left(\sqrt{\frac{1 - \rho}{1 + \rho}} t_2\right) - p_3 \Phi\left(\frac{-t_1 - \rho t_2}{\sqrt{1 - \rho^2}}\right) + p_4\Phi\left(\frac{t_1 + \rho t_2}{\sqrt{1 - \rho^2}}\right) + p_5\right)
     \end{align*}
 \end{proposition}

 We have the following second derivatives for $\Prob(\Theta_2, t_1, t_2)$.

 \begin{proposition}\label{prop:type4_hard_second_partials}
     We have
     \begin{align*}
     & \qquad \frac{\partial^2}{\partial t_1^2}\Prob(\Theta_2, t_1, t_2) \\  
     & = -t_1\varphi(t_1) \cdot \left(-2p_1  \Phi\left(\sqrt{\frac{1 - \rho}{1 + \rho}} t_1\right) + p_3 \Phi\left(\frac{t_2 + \rho t_1}{\sqrt{1 - \rho^2}}\right) - p_4\Phi\left(\frac{-t_2 - \rho t_1}{\sqrt{1 - \rho^2}}\right) + p_6\right) \\
     & \qquad +\varphi(t_1) 
     \cdot \left(-2p_1 \sqrt{\frac{1 - \rho}{1 + \rho}} \cdot \varphi\left(\sqrt{\frac{1 - \rho}{1 + \rho}} t_1\right) + (p_3 + p_4) \cdot \frac{\rho}{\sqrt{1 - \rho^2}}\cdot\varphi\left(\frac{t_2 + \rho t_1}{\sqrt{1 - \rho^2}}\right) \right),
     \end{align*}
     and
     \begin{align*}
     & \qquad \frac{\partial^2}{\partial t_2^2}\Prob(\Theta_2, t_1, t_2) \\ 
     & = -t_2 \varphi(t_2) \cdot \left(-2p_2  \Phi\left(\sqrt{\frac{1 - \rho}{1 + \rho}} t_2\right) - p_3 \Phi\left(\frac{-t_1 - \rho t_2}{\sqrt{1 - \rho^2}}\right) + p_4\Phi\left(\frac{t_1 + \rho t_2}{\sqrt{1 - \rho^2}}\right) + p_5\right) \\
     & \qquad + \varphi(t_2) \cdot\left(-2p_2  \sqrt{\frac{1 - \rho}{1 + \rho}}\cdot\varphi\left(\sqrt{\frac{1 - \rho}{1 + \rho}} t_2\right)+ (p_3 + p_4) \cdot \frac{\rho}{\sqrt{1 - \rho^2}}\cdot\varphi\left(\frac{t_1 + \rho t_2}{\sqrt{1 - \rho^2}}\right)\right),
     \end{align*}
     and 
     \[
     \frac{\partial^2}{\partial t_1 \partial t_2} \Prob(\Theta_2, t_1, t_2)  = (p_3 + p_4) \cdot \varphi_\rho(t_1, -t_2) = (p_3 + p_4) \cdot \frac{1}{\sqrt{1 - \rho^2}}\cdot\varphi(t_1)\varphi\left(\frac{t_2 + \rho t_1}{\sqrt{1 - \rho^2}}\right).
     \]
 \end{proposition}

 We would like to find the probabilities $p_1, \ldots, p_6$ that minimizes the maximum ratio achieved by any $\THRESH^-$ scheme $\max_{t_1, t_2} \Prob(\Theta_2, t_1, t_2) / \Value(\Theta_2)$. To do this, we will first heuristically derive a set of probabilities assuming $t_1 = -t_2$, and then verify that for these probabilities $\Prob(\Theta_2, t_1, t_2)$ is indeed maximized at a point where $t_1 = -t_2$.

\subsubsection{Deriving the probabilities}
 For $\Prob(\Theta_2, t, -t)$, we have
 \[
 \Prob(\Theta_2, t, -t) = (p_1 + p_4) \cdot (1 - \Phi_\rho(t, t)) + (p_2 + p_3) \cdot (1 - \Phi_\rho(-t, -t)) + p_5\Phi(-t) + p_6\Phi(t).
 \]
 We will choose $p_5 = p_6 = p$, which intuitively makes sense as $F_{\alpha^*}$ achieves the same ratio $\alpha^*$ on all $1$-configurations. Under this assumption we have
 \[
 \Prob(\Theta_2, t, -t) = (p_1 + p_4) \cdot (1 - \Phi_\rho(t, t)) + (p_2 + p_3) \cdot (1 - \Phi_\rho(-t, -t)) + p,
 \]
 and
 \[
 \frac{\partial }{\partial t}\Prob(\Theta_2, t, -t) = (p_1 + p_4) \cdot \left(-2\varphi(t)\cdot \Phi\left(\sqrt{\frac{1-\rho}{1+\rho}}t\right)\right) + (p_2 + p_3) \cdot \left(2\varphi(t)\cdot \Phi\left(-\sqrt{\frac{1-\rho}{1+\rho}}t\right)\right)
 \]

 Following the same strategy in Section~\ref{sec:type3-hardness}, we want the above to attain 0 at $t = t^* = \Phi^{-1}((1 - b^*)/2)$. This implies that
 \[
 \frac{p_2 + p_3}{p_1 + p_2 + p_3 + p_4} = \Phi\left(\sqrt{\frac{1 - \rho}{1 + \rho}} \cdot t^*\right) := r.
 \]

 This gives us the ratio between the probabilities of 2-configurations, we can then choose $p$ so that the two $\THRESH^-$ schemes $b \mapsto b$ and $b \mapsto 0$ in $F_{\alpha^*}$ achieves the same satisfying probability on $\Theta_2$, i.e.,
 \[
 \Prob(\Theta_2, t^*, -t^*) = (p_1 + p_4) \cdot (1 - \Phi_\rho(t^*, t^*)) + (p_2 + p_3) \cdot (1 - \Phi_\rho(-t^*, -t^*)) + p = p_1 + p_2 + p_3 + p_4.
 \]

 This implies that 
 \[
 \frac{p}{p_1 + p_2 + p_3 + p_4} = 1 - (1 - r) \cdot (1 - \Phi_\rho(t^*, t^*)) - r \cdot (1 - \Phi_\rho(-t^*, -t^*)):=r'.
 \]

 Since we also have $p_1 + p_2 + p_3 + p_4 + 2p = 1$, the above gives
 \[
     p = \frac{2r'}{1 + 2r'}, \quad  p_1 + p_4  \EQ (1-r) \cdot (1 - 2p), \quad p_2 + p_3  \EQ r \cdot (1 - 2p)\;.\\
 \]

Finally, by setting the partial derivatives $\frac{\partial}{\partial t_1}\Prob(\Theta_2, t^*, -t^*) = \frac{\partial}{\partial t_2}\Prob(\Theta_2, t^*, -t^*) = 0$, we obtain that $p_1 = p_2 = p$ as well. In summary, we obtain the following probabilities:
 \[
     p_1 = p_2 = p_5 = p_6 = p = \frac{2r'}{1 + 2r'}, \quad p_3  = r \cdot (1 - 2p) - p, \quad p_4  = (1-r) \cdot (1 - 2p) - p.\\
 \]

 The numeric values for these probabilities is listed as follows:
 \[
     p_1 = p_2 = p_5 = p_6 \approx 0.0858, \quad p_3 \approx 0.1737, \quad p_4 \approx 0.4831.
 \]

We remark that since we have chosen the hardest bias $b^*$, the approximation ratio achieved by $F_{\alpha^*}$ on this distribution is exactly $\alpha^*$. In fact, by design, both functions in $F_{\alpha^*}$ achieve exactly $\alpha^* = \frac{\Prob(\Theta_2, -\infty, -\infty)}{\Value(\Theta_2)} = \frac{\Prob(\Theta_2, t^*, -t^*)}{\Value(\Theta_2)}$.

 \subsubsection{Verifying that $t_1 = -t_2$ at a global maximum}

Now we prove that with the probabilities computed in the previous section,  $(t^*, -t^*)\in \mathbb{R}^2$ is indeed a global maximum for $\Prob(\Theta_2, t_1, t_2)$. To give a better sense of the function that we are working with, we give the following plot. 

\begin{figure}[H]
\begin{center}
\includegraphics[width=2.8in]{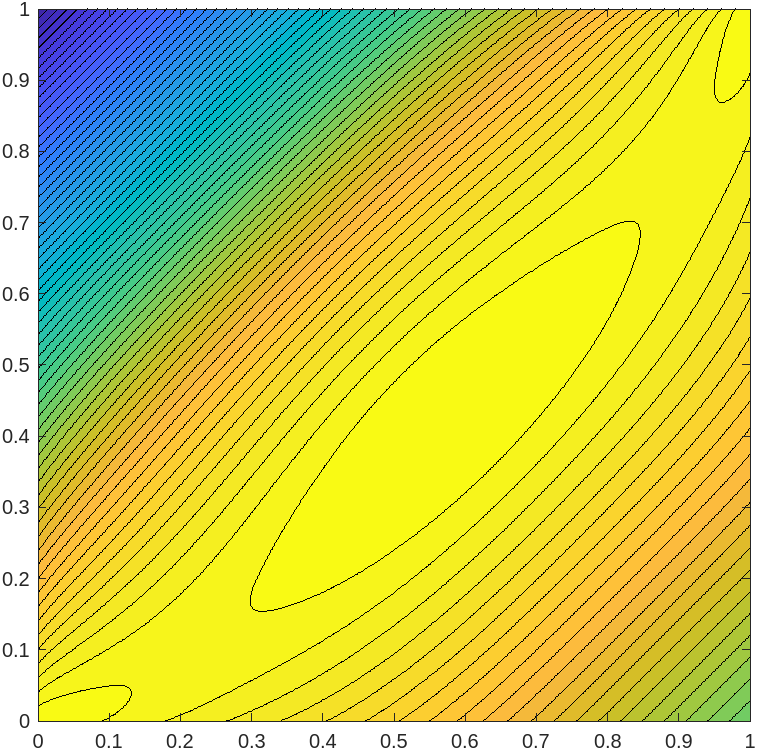}
\end{center}
\caption{A contour plot of $\Prob(\Theta_2, t_1, t_2)$, where the x-axis is $\Phi(t_1)$ and the y-axis is $\Phi(t_2)$}
\end{figure}

It can be seen that aside from $(t^*, -t^*)$, there are two other critical points which are saddle points. This creates complications for an analytic proof. We will circumvent this difficulty by employing interval arithmetic. We first prove the following statement with interval arithmetic.

\begin{lemma}[Interval arithmetic]\label{lem:type4_hard_ia}
For every $t_1, t_2 \in \mathbb{R}$, at least one of the following is true:
\begin{itemize}
    \item $\Prob(\Theta_2, t_1, t_2) < \Prob(\Theta_2, t^*, -t^*)$.
    \item $t_1, t_2 \leq \Phi^{-1}(0.0001)$ or $t_1, t_2 \geq \Phi^{-1}(0.9999)$.
    \item $|t_1 - t^*|, |t_2 + t^*| < 0.01$ and the Hessian matrix for $\Prob(\Theta_2, t_1, t_2)$ is negative definite.
    \item $\frac{\partial}{\partial t_1}\Prob(\Theta_2, t_1, t_2) \neq 0$ or $\frac{\partial}{\partial t_2}\Prob(\Theta_2, t_1, t_2) \neq 0$
\end{itemize}
\end{lemma}

Since the gradient of $\Prob(\Theta_2, t_1, t_2)$ vanishes at $(t^*, -t^*)$, the third item shows that $\Prob(\Theta_2, t_1, t_2) \leq \Prob(\Theta_2, t^*, -t^*)$ for every $t_1 \in [t^* - 0.01, t^* + 0.01], t_2 \in [-t^* - 0.01, -t^* + 0.01]$.

The following proposition deals with the boundary situation that our interval arithmetic does not certify directly.
\begin{proposition}\label{prop:type4_hard_infbound}
    Let $t_1, t_2 \in \mathbb{R}^2$ be such that $t_1, t_2 \leq \Phi^{-1}(0.0001)$ or $t_1, t_2 \geq \Phi^{-1}(0.9999)$, then we have $\frac{\partial}{\partial t_1}\Prob(\Theta_2, t_1, t_2) \neq 0$ or $\frac{\partial}{\partial t_2}\Prob(\Theta_2, t_1, t_2) \neq 0$.
\end{proposition}
\begin{proof}
Assume that $t_1, t_2 \leq \Phi^{-1}(0.0001)$.

     Since $\varphi(t_1), \varphi(t_2) > 0$, the partial derivatives being 0 is equivalent to
     \begin{align*}
         -2p  \Phi\left(\sqrt{\frac{1 - \rho}{1 + \rho}} t_1\right) + p_3 \Phi\left(\frac{t_2 + \rho t_1}{\sqrt{1 - \rho^2}}\right) - p_4\Phi\left(\frac{-t_2 - \rho t_1}{\sqrt{1 - \rho^2}}\right) + p & = 0, \\
         -2p  \Phi\left(\sqrt{\frac{1 - \rho}{1 + \rho}} t_2\right) - p_3 \Phi\left(\frac{-t_1 - \rho t_2}{\sqrt{1 - \rho^2}}\right) + p_4\Phi\left(\frac{t_1 + \rho t_2}{\sqrt{1 - \rho^2}}\right) + p & = 0.
     \end{align*}
     Using the fact that $\Phi(x) = 1 - \Phi(-x)$, we can rewrite the above as
     \begin{align*}
         -2p  \Phi\left(\sqrt{\frac{1 - \rho}{1 + \rho}} t_1\right) + (p_3 + p_4) \Phi\left(\frac{t_2 + \rho t_1}{\sqrt{1 - \rho^2}}\right) - p_4 + p & = 0, \\
         -2p  \Phi\left(\sqrt{\frac{1 - \rho}{1 + \rho}} t_2\right)  + (p_3 + p_4)\Phi\left(\frac{t_1 + \rho t_2}{\sqrt{1 - \rho^2}}\right) - p_3 + p & = 0.
     \end{align*}
     Since $\frac{t_2 + \rho t_1}{\sqrt{1 - \rho^2}} + \frac{t_1 + \rho t_2}{\sqrt{1 - \rho^2}} = \frac{(t_1 + t_2)\sqrt{1 + \rho}}{\sqrt{1 - \rho}}$, we have either $\frac{t_2 + \rho t_1}{\sqrt{1 - \rho^2}} \leq \frac{(t_1 + t_2)\sqrt{1 + \rho}}{2\sqrt{1 - \rho}}$ or  $\frac{t_1 + \rho t_2}{\sqrt{1 - \rho^2}} \leq \frac{(t_1 + t_2)\sqrt{1 + \rho}}{2\sqrt{1 - \rho}}$. A simple estimation shows that we would then have either $\Phi\left(\frac{t_2 + \rho t_1}{\sqrt{1 - \rho^2}}\right) \leq \Phi(\frac{(t_1 + t_2)\sqrt{1 + \rho}}{2\sqrt{1 - \rho}}) < 0.12 $ or  $\Phi\left(\frac{t_1 + \rho t_2}{\sqrt{1 - \rho^2}}\right) \leq \Phi(\frac{(t_1 + t_2)\sqrt{1 + \rho}}{2\sqrt{1 - \rho}}) < 0.12$. But in either case we would violate at least one of the equations, since we'd have either
     \[
      -2p  \Phi\left(\sqrt{\frac{1 - \rho}{1 + \rho}} t_1\right) + (p_3 + p_4) \Phi\left(\frac{t_2 + \rho t_1}{\sqrt{1 - \rho^2}}\right) - p_4 + p < (p_3 + p_4) \cdot 0.12 - p_4 + p < 0
     \]
     or
     \[
      -2p  \Phi\left(\sqrt{\frac{1 - \rho}{1 + \rho}} t_2\right)  + (p_3 + p_4)\Phi\left(\frac{t_1 + \rho t_2}{\sqrt{1 - \rho^2}}\right) - p_3 + p < (p_3 + p_4) \cdot 0.12 - p_3 + p < 0.
     \]
     This shows that at least one of the partial derivatives is non-zero. The case where $t_1, t_2 \geq \Phi^{-1}(0.9999)$ can be dealt with similarly.
\end{proof}

 \begin{proposition}
     For every $t_1, t_2 \in \mathbb{R} \cup \{\pm \infty\}$, we have
     $\Prob(\Theta_2, t_1, t_2) \leq \Prob(\Theta_2, t^*, -t^*)$.
 \end{proposition}
 \begin{proof}
     We first consider the infinity cases. If $t_1 = -\infty$, then we have
     \begin{align*}
         \Prob(\Theta_2, t_1, t_2) = & \,\, p_1 \cdot (1 - \Phi_\rho(t_1, t_1)) + p_2 \cdot (1 - \Phi_\rho(t_2, t_2)) + p_3 \cdot (1 - \Phi_\rho(-t_1, t_2)) \\
     &+ p_4 \cdot (1 - \Phi_\rho(-t_2, t_1)) + p_5 \cdot \Phi(t_2) + p_6 \cdot \Phi(t_1)\\
     = & \,\, p_1 + p_2 \cdot (1 - \Phi_\rho(t_2, t_2)) + p_3 \cdot (1 - \Phi(t_2)) + p_4 + p_5 \cdot \Phi(t_2).
     \end{align*}
     Since $p_3 > p_5$, the $\Prob(\Theta_2, -\infty, t_2)$ is monotonically decreasing in $t_2$, so we should choose $t_2 = -\infty$ as well. A similar analysis shows that if $t_1 = +\infty$, then we should also set $t_2 = +\infty$, and furthermore $\Prob(\Theta_2, -\infty, -\infty) = \Prob(\Theta_2, \infty, \infty)$.
     
     Now assume that there is a global maximum $(t_1, t_2)$ with $\Prob(\Theta_2, t_1, t_2) > \Prob(\Theta_2, t^*, -t^*) = \Prob(\Theta_2, -\infty, -\infty)$. Since $\Prob(\Theta_2, t_1, t_2)$ is a smooth function, the gradient vanishes at the global maximum, so by Lemma~\ref{lem:type4_hard_ia} and Proposition~\ref{prop:type4_hard_infbound} we must have $|t_1 - t^*|, |t_2 + t^*| < 0.001$. However, the negative definiteness of the Hessian matrix in that neighborhood would then imply that $\Prob(\Theta_2, t_1, t_2) \leq \Prob(\Theta_2, t^*, -t^*)$. This contradiction shows that there is no global maximum strictly larger than $\Prob(\Theta_2, t^*, -t^*)$, and therefore $(t^*, -t^*)$ itself must be a global maximum.
 \end{proof}
 
The above analysis combined with Theorem~\ref{thm:PCP} immediately implies the following theorem.

\begin{theorem}\label{thm:type4-hard}
    For every $\epsilon > 0$, it is UG-hard to approximate MAX $\{1,2\}$-HORN SAT within a factor of $\alpha^* + \epsilon$.
\end{theorem}

\section{Concluding remarks and open problems}\label{S-concl}

We have proved Austrin's \cite{Austrin07} simplicity conjecture, thereby determining the optimal approximation ratio of \MAXSAT{2}, modulo only UGC. We have also obtained a complete classification, in terms of their optimal approximation ratios, for all subproblems of \MAXSAT{2}. We introduced two interesting non-trivial subproblems of \MAXSAT{2}, namely \MAXHORNSAT{$\{1,2\}$} and \MAXCSPF{$\{x\lor y,x,\bar{x}\}$}, for which larger approximation ratios can be obtained.

Our proof of the simplicity conjecture used a combination of analytic and rigorous computational tools, namely interval arithmetic. Although interval arithmetic was used before to analyze SDP-based approximation algorithms, our use is much more involved than previous uses, since we were not just trying to lower bound the approximation ratio achieved by the algorithm, but rather certify that the exact worst-case behavior of the algorithm is obtained on configurations of a certain form, the so-called simple configurations. This requires a much more careful analysis.

It would be interesting to see whether some of the computational parts of our proof of the simplicity conjecture can be replaced by purely analytical arguments. This may help obtaining tight approximability results for other, more complicated, \MAXCSPP\ problems.

Raghavendra and Tan \cite{RaTa12} obtained approximation algorithms for the \MAXCUT\ and \MAXSAT{2} problems with a global \emph{cardinality constraint}. An interesting special case of \MAXSAT{2} with a cardinality condition is \MAXkVC, the problem of choosing~$k$ vertices of an undirected graph so as to maximize the number of covered edges. Austrin and Stankovi\'{c} \cite{AuSt19} (see also \cite{Stankovic23}) showed that the algorithms of Raghavendra and Tan are optimal, modulo UGC and modulo an appropriate version of the simplicity conjecture. We believe that the techniques we used here can also be used to prove this conjecture. We hope to include a proof in the next version of this paper.

Another subject worth exploring is the striking difference between \MAXSAT{2} and its subproblems, and \MAXAND\ and its subproblem \MAXDICUT. \MAXSAT{2} and its subproblems have optimal approximation functions that use very simple and natural threshold functions. On the other hand, no optimal algorithms for \MAXAND\ and \MAXDICUT\ are known, or conjectured, and very complicated and non-intuitive threshold functions are needed to obtain close to optimal approximation algorithms for them~\cite{BHPZ22}. It would be interesting to understand the reason for this difference.

\subsection*{Acknowledgments}

We thank Per Austrin for helpful discussions, including confirming that his Simplicity Conjecture was still open. We thank Aaron Potechin for helpful discussions during the pursuit of this project. We thank anonymous reviewers for numerous helpful comments. We especially thank one reviewer who pointed out an inaccuracy in the initial phrasing and proof of Theorem~\ref{thm:Austrin}. We thank Rohit Agarwal and Venkatesan Guruswami for pointing out an issue in an earlier version with the code link.

\bibliography{refs.bib}
\bibliographystyle{plain}

\newpage

\appendix

\section{Interval Arithmetic}\label{app:IA}

As mentioned earlier, we use the C library Arb~\cite{johansson2017arb} for our interval arithmetic. As the assembly-like syntax of Arb is a bit tedious for writing long programs, we wrote\footnote{Code is publicly available at: \url{https://github.com/jbrakensiek/max2sat}} a C++ wrapper\footnote{This wrapper does not expose the full functionality of Arb as it does not allow for setting the individual precision of each computation. Instead, we use a fixed global precision of 64 bits.} for the Arb library functions we needed for our verifications. Implementing the various mathematical expressions which appear in the paper are rather straightforward, with the exception of the function $\Phi_{\rho}(t_1, t_2)$. In that case, we adapt the integration formula used by \cite{BHPZ22} and inspired from \cite{DW90}. We also adapt other portions of the code used by \cite{BHPZ22} within our wrapper. We also take care to properly handle boundary conditions such as $\rho = \pm 1$ or $t_1, t_2 = \pm \infty$. 

\subsection{General divide-and-conquer technique}

We now outline the general divide-and-conquer procedure we use to verify propositions in interval arithmetic. Similar techniques are used in earlier works.

\renewcommand{\check}{\operatorname{check}}

Given a parameter such as $\rho \in \RR$, we let $\bar{\rho}$ denote an interval of possible values of $\rho$. Given intervals $\bar{x}_1, \hdots, \bar{x}_n$ and \emph{criteria}\footnote{These can be thought of as subsets of $\RR^n$ where the criteria holds.} $C_1, \hdots, C_k$, we in general run the following divide-and-conquer procedure, which we call $\check(\bar{x}_1, \hdots, \bar{x}_n; C_1, \hdots, C_k)$.

\begin{itemize}
\item If for some $i \in [k]$, $\bar{x}_1 \times \cdots \times \bar{x}_n$ satisfies $C_i$, halt.

\item Otherwise, pick $i \in [n]$ for which $\bar{x}_i$ is of shortest (but nonzero) length, and call \[\check(\bar{x}_1, \hdots, \bar{x}_{i-1}, \bar{x}'_i, \bar{x}_{i+1}, \hdots,  \bar{x}_n; C_1, \hdots, C_k)\] and
\[\check(\bar{x}_1, \hdots, \bar{x}_{i-1}, \bar{x}''_i, \bar{x}_{i+1}, \hdots,  \bar{x}_n; C_1, \hdots, C_k),\]
where $\bar{x}'_{i} \cup \bar{x}''_i$ divide $\bar{x}_i$ into halves.
\end{itemize}

Note that if $\check$ terminates then, we know for every $(x_1, \hdots, x_n) \in \bar{x}_1 \times \cdots \times \bar{x}_n$, at least one of $C_1, \hdots, C_k$ holds. 

Each of the lemmas proved with interval arithmetic in this paper are done by selecting suitable criteria $C_1, \hdots, C_k$, and then running the correspond interval arithmetic program until successful termination. We note that every theorem in this paper (in total) can be verified in less than 15 minutes on a mid-range desktop computer.

\subsection{Implementation specifics for each theorem}

In this section, we discuss the specific call(s) to $\check$ which are implemented to verify the interval arithmetic lemmas in the main body of the paper.

\subsubsection{Lemma~\ref{lem:new_step1}}

Let $g_{\beta}(b_1, b_2, \rho) = f_{\beta}(b_1, b_2, b_{12}(b_1,b_2,\rho)) $, where $b_{12}(b_1, b_2, \rho) = b_1b_2 + \rho \sqrt{(1-b_1^2)(1-b_2)^2}$.

Here we run $\check$ on variables $\bar{b}_1 = \bar{b}_2 = \bar{\rho} = [-1, 1]$ and $\bar{\beta} = [0.94, 0.9405]$ (which is a bit larger than the claimed interval). We then have three criteria $C_1, C_2, C_3$ which check the following:
\begin{itemize}
    \item $C_1$: $b_{12}(b_1, b_2, \rho) < -1 + |b_1 + b_2|$.
    \item $C_2$: $g_{\beta}(b_1, b_2, \rho) > 0.001$
    \item $C_3$: $b_{12}(b_1, b_2, \rho) < 1 - |b_1 - b_2|$ and $\nabla g_{\beta} \neq (0, 0, 0)$ 
\end{itemize}   

Each of these criteria is straightforward to implement in interval arithmetic, and we verified that $\check(\bar{b}_1, \bar{b}_2, \bar{\rho}, \bar{\beta} ; C_1, C_2, C_3 )$ halts in less than 2 minutes on a desktop computer.

Consider any $(b_1, b_2, b_{12}) \in [-1, 1]^3$ for which there is at least one $\rho \in [-1, 1]$ with $b_{12} = b_{12}(b_1, b_2, \rho)$. Thus, as certified by $\check$, we either have that $b_{12} < -1 + |b_1 + b_2|$, $b_{12} < 1 - |b_1 - b_2|$ and $g_{\beta}(b_1, b_2, \rho) = f_{\beta}(b_1, b_2, b_{12}) > 0.001$, or $\nabla g_{\beta}(b_1, b_2, \rho) \neq (0, 0, 0)$. The first two directly correspond to cases of Lemma~\ref{lem:new_step1}. In the last case, observe that
\begin{align*}
    \frac{\partial g_{\beta}}{\partial b_1} &= \frac{\partial f_{\beta}}{\partial b_1} + \frac{\partial b_{12}}{\partial b_1} \frac{\partial f_{\beta}}{\partial b_{12}}\\
    \frac{\partial g_{\beta}}{\partial b_2} &= \frac{\partial f_{\beta}}{\partial b_2} + \frac{\partial b_{12}}{\partial b_2} \frac{\partial f_{\beta}}{\partial b_{12}}\\
    \frac{\partial g_{\beta}}{\partial \rho} &=  \frac{\partial b_{12}}{\partial \rho} \frac{\partial f_{\beta}}{\partial b_{12}}
\end{align*}
Thus, $\nabla f_{\beta}(b_1, b_2, b_{12})$ must be also nonzero, as desired. Thus, Lemma~\ref{lem:new_step1} holds. 

\subsubsection{Lemma~\ref{lem:new_step2}}

Let $h_{\beta}(b_1, b_2) = f_{\beta}(b_1, b_2, -1 + b_1 + b_2)$. 

Here we run $\check$ on variables $\bar{b}_1 = \bar{b}_2 = [-1, 1]$ and $\bar{\beta} = [0.9401653, 9401658]$. We then have four criteria which check:
\begin{itemize}
    \item $C_1$: $h_{\beta}(b_1, b_2) > 0.001$
    \item $C_2$: $b_1 + b_2 < 0$
    \item $C_3$: $b_1 + b_2 > 0$ and $\nabla h_{\beta}(b_1, b_2) \neq (0, 0)$
    \item $C_4$: $b_1, b_2 \in [b_0 - \eps, b_0 + \eps]$.
\end{itemize}   

We verified that $\check(\bar{b}_1, \bar{b}_2, \bar{\beta} ; C_1, C_2, C_3, C_4 )$ halts in less than 2 minutes on a desktop computer. Note that criteria $C_2$ has implies that at least one of $C_1, C_3, C_4$ hold for all $b_1, b_2 \in [-1, 1]$ with $b_1 + b_2 \ge 0$. Thus, the primary claim of Lemma~\ref{lem:new_step2} holds.

For the secondary claims, the $f_{0.9401658}(b_0, b_0, -1 + 2b_0) < 0$ is a direct computation. The claim that for all $b_1, b_2 \in [b_0 - \eps, b_0 + \eps]$, we have that $f_{0.9401653}(b_1, b_2, -1 + b_1 + b_2) > 0$ is checked by running $\check(\bar{b}_1, \bar{b}_2 ; C_5)$, where $\bar{b}_1 = \bar{b}_2 = [b_0 - \eps, b_0 + \eps]$ and $C_5$ is the criteria that $f_{0.9401653}(b_1, b_2, -1 + b_1 + b_2) > 0$. This $\check$ also terminates in under a minute.

\subsubsection{Lemma~\ref{lem:type3-step1}}

Let $\ell_{\gamma}(b_1, b_2, \rho) = h_{\gamma}(b_1, b_2, b_{12}(b_1,b_2,\rho)) $, where $b_{12}(b_1, b_2, \rho) = b_1b_2 + \rho \sqrt{(1-b_1^2)(1-b_2)^2}$.

Here we run $\check$ on variables $\bar{b}_1 = \bar{b}_2 = \bar{\rho} = [-1, 1]$ and $\bar{\gamma} = [0.95, 0.96]$. We then have three criteria $C_1, C_2, C_3$ which check the following:
\begin{itemize}
    \item $C_1$: $b_{12}(b_1, b_2, \rho) < -1 + |b_1 + b_2|$.
    \item $C_2$: $\ell_{\gamma}(b_1, b_2, \rho) > 0.001$
    \item $C_3$: $b_{12}(b_1, b_2, \rho) < 1 - |b_1 - b_2|$ and $\nabla \ell_{\gamma} \neq (0, 0, 0)$ 
\end{itemize}   

Each of these criteria is straightforward to implement in interval arithmetic, and we verified that $\check(\bar{b}_1, \bar{b}_2, \bar{\rho}, \bar{\gamma} ; C_1, C_2, C_3 )$ halts in less than 2 minutes on a desktop computer.

Consider any $(b_1, b_2, b_{12}) \in [-1, 1]^3$ for which there is at least one $\rho \in [-1, 1]$ with $b_{12} = b_{12}(b_1, b_2, \rho)$. Thus, as certified by $\check$, we either have that $b_{12} < -1 + |b_1 + b_2|$, $b_{12} < 1 - |b_1 - b_2|$ and $\ell_{\gamma}(b_1, b_2, \rho) = h_{\gamma}(b_1, b_2, b_{12}) > 0.001$, or $\nabla \ell_{\gamma}(b_1, b_2, \rho) \neq (0, 0, 0)$. The first two directly correspond to cases of Lemma~\ref{lem:type3-step1}. In the last case, observe by the same logic as Lemma~\ref{lem:new_step1}, we have that $\nabla \ell_{\gamma}(b_1, b_2, \rho) \neq 0$ implies that $\nabla f_{\beta}(b_1, b_2, b_{12})$ is nonzero, as desired. Thus, Lemma~\ref{lem:type3-step1} holds.

For every $b_1, b_2, b_{12} \in [-1,1]$ and $\gamma \in [0.95, 0.96]$, at least one of the following is true:
\begin{itemize}
    \item $b_{12} < -1 + |b_1 + b_2|$,
    \item $h_\gamma(b_1, b_2, b_{12}) > 0.001$,
    \item $b_{12} < 1 - |b_1 - b_2|$ and $\nabla h_\gamma \neq (0, 0, 0)$,
    \item $\rho(b_1, b_2, b_{12}) \notin [-1, 1]$.
\end{itemize}

\subsubsection{Lemma~\ref{lem:type3-step2}}

Let $h^{\Delta}_{\gamma}(b_1, b_2) = h_{\gamma}(b_1, b_2, -1 + b_1 + b_2)$. 

Here we run $\check$ on variables $\bar{b}_1 = \bar{b}_2 = [-1, 1]$ and $\bar{\gamma} = [0.9539798, 0.95398]$. We then have three criteria to check:
\begin{itemize}
    \item $C_1$: $h^{\Delta}_{\gamma}(b_1, b_2) > 0.001$
    \item $C_2$: $b_1 + b_2 < 0$ and $\nabla h^{\delta}_{\gamma}(b_1, b_2) \neq (0, 0)$
    \item $C_3$: $b_1, b_2 \in [b_0 - \eps, b_0 + \eps]$.
\end{itemize}   

We verified that $\check(\bar{b}_1, \bar{b}_2, \bar{\gamma} ; C_1, C_2, C_3)$ halts in less than 2 minutes on a desktop computer. Thus, Lemma~\ref{lem:new_step2} holds.

For the secondary claims, the assertion that $h_{0.95398}(b_0, b_0, -1 - 2b_0) < 0$ is checked by direct computation. For all $b_1, b_2 \in [b_0 - \eps, b_0 + \eps]$, the claim that $h_{0.9539798}(b_1, b_2, -1 - b_1 - b_2) > 0$ is checked precisely as in Lemma~\ref{lem:new_step2}, which terminated in under a minute.

\subsubsection{Lemma~\ref{lem:type4-step1}}

Let $g(b_1, b_2, \rho) = f(b_1, b_2, b_{12}(b_1,b_2,\rho)) $, where $b_{12}(b_1, b_2, \rho) = b_1b_2 + \rho \sqrt{(1-b_1^2)(1-b_2)^2}$.

Here we run $\check$ on variables $\bar{b}_1 = \bar{b}_2 = \bar{\rho} = [-1, 1]$. We then have three criteria $C_1, C_2, C_3$ which check the following:
\begin{itemize}
    \item $C_1$: $b_{12}(b_1, b_2, \rho) < -1 + |b_1 + b_2|$.
    \item $C_2$: $b_{12} < -1 + |b_1 + b_2|$ and $g(b_1, b_2, \rho) > 1 - \frac{1}{0.95}$
    \item $C_3$: $\nabla g\neq (0, 0, 0)$ 
\end{itemize}   

Each of these criteria is straightforward to implement in interval arithmetic, and we verified that $\check(\bar{b}_1, \bar{b}_2, \bar{\rho} ; C_1, C_2, C_3 )$ halts in less than 1 minute on a desktop computer.

Consider any $(b_1, b_2, b_{12}) \in [-1, 1]^3$ for which there is at least one $\rho \in [-1, 1]$ with $b_{12} = b_{12}(b_1, b_2, \rho)$. Thus, as certified by $\check$, we either have that $b_{12} < -1 + |b_1 + b_2|$, $b_{12} < -1 + |b_1 + b_2|$ and $g(b_1, b_2, \rho) = f(b_1, b_2, b_{12}) > 1 - \frac{1}{0.95}$, or $\nabla g(b_1, b_2, \rho) \neq (0, 0, 0)$. The first two directly correspond to cases of Lemma~\ref{lem:type4-step1}. In the last case, In the last case, observe by the same logic as Lemma~\ref{lem:new_step1}, $\nabla g(b_1, b_2, \rho) \neq 0$ implies that $\nabla f(b_1, b_2, b_{12})$ must be also nonzero, as desired. Thus, Lemma~\ref{lem:type4-step1} holds.

\subsubsection{Lemma~\ref{lem:type4-step2}}

Let $h(b_1, b_2) = f(b_1, b_2, -1 + b_1 + b_2)$. 

Here we run $\check$ on variables $\bar{b}_1 = \bar{b}_2 = [-1, 1]$. We then have four criteria which check:
\begin{itemize}
    \item $C_1$: $h(b_1, b_2) > 0.001$
    \item $C_2$: $b_1 + b_2 < 0$
    \item $C_3$: $b_1 + b_2 > 0$ and $\nabla h(b_1, b_2) \neq (0, 0)$
    \item $C_4$: $b_1, b_2 \in [b_0 - \eps, b_0 + \eps]$.
\end{itemize}   

We verified that $\check(\bar{b}_1, \bar{b}_2, \bar{\beta} ; C_1, C_2, C_3, C_4 )$ halts in less than 2 minutes on a desktop computer. Note that criteria $C_2$ has implies that at least one of $C_1, C_3, C_4$ hold for all $b_1, b_2 \in [-1, 1]$ with $b_1 + b_2 \ge 0$. Thus, Lemma~\ref{lem:new_step2} holds. Note that the additional claim that $f(b_0, b_0, -1 + 2b_0) < 1 - \frac{1}{0.9462}$ is checked with a direct computation.

\subsubsection{Lemma~\ref{lem:type4_hard_ia}}

Let $s_{b}(t_1, t_2) = \Prob(\Theta_2, t_1, t_2)$.

We set $\tau_1, \tau_2$ represent normalized thresholds with respect to $[0, 1]$. That is, $t_1 = \Phi^{-1}(\tau_1)$, etc.

Here we run $\check$ on variables $\bar{\tau}_1 = \bar{\tau}_2 = [0, 1]$ and $\bar{b} = [b_0 - \eps, b_0 + \eps]$. We then have four criteria which check:

\begin{itemize}
    \item $C_1$: $s_b(\Phi^{-1}(\tau_1),\Phi^{-1}(\tau_2)) < s_b(\Phi^{-1}((1-b)/2), \Phi^{-1}((1+b)/2))$.
    \item $C_2$: $\tau_1, \tau_2 < 10^{-4}$ or $\tau_1, \tau_2 > 1 - 10^{-4}$.
    \item $C_3$: $\tau_1, \tau_2$ within $10^{-2}$ of $((1-b)/2, (1+b)/2)$ in each coordinate and the hessian matrix for $\Prob(\Theta_2, t_1, t_2)$ evaluated at ${(\Phi^{-1}(\tau_1), \Phi^{-1}(\tau_2))}$ is negative definite.
    \item $C_4$: $\nabla s_b|_{(\Phi^{-1}(\tau_1), \Phi^{-1}(\tau_2))} \neq (0, 0)$.
\end{itemize}   

We verified that $\check(\bar{\tau}_1, \bar{\tau}_2, \bar{b} ; C_1, C_2, C_3, C_4 )$ halts in less than 15 minutes on a desktop computer.

\section{Explicit equations for $\beta_{LLZ}$}\label{app:beta}

The optimal approximation ratio $\beta=\beta_{LLZ}\approx 0.940165$ of \MAXSAT{2}, under UGC, does not seem to have a simple explicit form. This is probably not so surprising as $\alpha=\alpha_{GW}\approx 0.878567$, the optimal approximation ratio for \MAXCUT, is also not given explicitly, though it can be expressed as $\alpha_{GW}=\frac{2}{\pi \sin\theta}$, where $\theta\approx 2.33$ is the solution of the equation $\theta=\tan(\theta/2)$. (See, e.g., \cite{KKMO07}.)

In this section we try to give a simple set of equations that express $\beta_{LLZ}$. The obtained equations are somewhat more complicated but are still conceptually simple, two (nonlinear) equations in two real variables.

Let $P_\beta(-b)=\Prob_\beta(-b,-b,-1+2b)$, where $0<b<1$. It is not difficult to see that the optimal~$b$ and~$\beta$ should satisfy the two equations $P_\beta(-b)=\beta$ and $P'_\beta(-b)=0$, where the derivative is with respect to $b$. This gives us two equations in the two unknowns $b$ and $\beta$.

More explicitly, We have the following expressions for $P_\beta(-b)$ and $P'_\beta(-b)$:
\[P_\beta(-b) \EQ 1-\Phi_\rho(t,t)\;,\; \text{ where } \rho=\frac{b-1}{b+1} \;,\; t=\Phi^{-1}\left(\frac{1-\beta b}{2}\right) \;,\]
\[P'_\beta(-b) \EQ -\frac{{\rm e}^{-\frac{1+b}{2b}t^2}}{2\pi\sqrt{b}(1+b)} + \beta \Phi\left(\frac{t}{\sqrt{b}}\right)\;.\]
The formula for $P'_\beta(-b)$ is obtained by applying the chain rule, using the known partial derivatives of $\Phi_\rho(x,y)$. All together, we have the following two equations expressing $b$ and $\beta$, with the help of the two auxiliary variables $\rho$ and $t$:
\[ 1-\Phi_\rho(t,t) \EQ \beta \EQ \frac{{\rm e}^{-\frac{1+b}{2b}t^2}}{2\pi\sqrt{b}(1+b) \Phi(\frac{t}{\sqrt{b}})}\;,\; \text{ where } \rho=\frac{b-1}{b+1} \;,\; t=\Phi^{-1}\left(\frac{1-\beta b}{2}\right) \;.\]

Solving these two equations numerically we get:
\[\beta \approx 0.9401656724814047324615850917696020973303754687978028584668520377\;,\]
\[b \approx 0.1624783228980762946610658853055298253055592890270700849334606721\;.\]

An equivalent formulation, obtained by treating $b$ and~$t$ as the two independent variables, and $\beta$ as an auxiliary variable, is the following:

\[ \beta \EQ 1-\Phi_{\frac{b-1}{b+1}}(t,t) \EQ \frac{1-2\Phi(t)}{b} \EQ \frac{{\rm e}^{-\frac{1+b}{2b}t^2}}{2\pi\sqrt{b}(1+b) \Phi(\frac{t}{\sqrt{b}})} \;.\]

\end{document}